\documentclass[letterpaper,twocolumn,10pt]{article}
\usepackage{usenix2019_v3}

\usepackage{tikz}
\usepackage{amsmath}

\usepackage{amssymb}
\usepackage{amsthm}
\usepackage{algorithmic}

\begin{document}

\date{}

\title{A Highly Accurate Query-Recovery Attack against\\ Searchable Encryption using Non-Indexed Documents}

\author{
	{\rm Marc Damie}\\
	University of Technology of Compiègne, France
	\and
	{\rm Florian Hahn}\\
	University of Twente, The Netherlands
	\and
	{\rm Andreas Peter}\\
	University of Twente, The Netherlands
} 

\maketitle

\begin{abstract}
	Cloud data storage solutions offer customers cost-effective and reduced data management. While attractive, data security issues remain to be a core concern. Traditional encryption protects stored documents, but hinders simple functionalities such as keyword search. Therefore, searchable encryption schemes have been proposed to allow for the search on encrypted data. Efficient schemes leak at least the access pattern (the accessed documents per keyword search), which is known to be exploitable in query recovery attacks assuming the attacker has a significant amount of background knowledge on the stored documents. Existing attacks can only achieve decent results with strong adversary models (e.g. at least 20\% of previously known documents or require additional knowledge such as on query frequencies) and they give no metric to evaluate the certainty of recovered queries. This hampers their practical utility and questions their relevance in the real-world.

	We propose a refined score attack which achieves query recovery rates of around 85\% without requiring exact background knowledge on stored documents; a distributionally similar, but otherwise different (i.e., non-indexed), dataset suffices. The attack starts with very few known queries (around 10 known queries in our experiments over different datasets of varying size) and then iteratively recovers further queries with confidence scores by adding previously recovered queries that had high confidence scores to the set of known queries.
	Additional to high recovery rates, our approach yields interpretable results in terms of confidence scores.
\end{abstract}


\section{Introduction}
\label{sec:introduction}
Cloud data storage services continue to be on the rise and attract more users than ever before. At the same time, data is a major target in cyber-attacks and the headlines about data breaches have become mainstream. This makes data security a key concern in this setting. While traditional encryption technology can be used to protect the confidentiality of data, simple functionalities such as searching get lost under encryption. To cope with this, Song, Wagner, and Perrig \cite{songpractical} presented a practical solution to search on encrypted data. A few years later, Curtmola et al. \cite{Curtmola} presented their construction of a searchable symmetric encryption (SSE) scheme based on an inverted index. As a result, their construction can search keywords in encrypted documents in optimal search time.

An (index-based) SSE scheme creates an encrypted index that can be queried to obtain the documents' identifiers containing one keyword.
The encrypted index hides the underlying keywords from the server but leaks the {\it access pattern} for each query; the access pattern is the list of identifiers of all documents containing the queried keyword. This work focuses on single-keyword search SSE schemes that leak the access pattern; we do not consider more complex systems such as encrypted databases.

The access pattern leakage has been shown to be exploitable in passive attacks. Blackstone et al.~\cite{blackstonerevisiting} divided such passive attacks into two categories: 1)~known-data attacks where the adversary has partial or complete knowledge of the documents indexed by the server (also referred to as leakage-abuse attacks), and 2) similar-data attacks where the adversary only knows some (non-indexed) documents similar to the indexed documents (also referred to as inference attacks).
An adversary who can run known-data attacks is more powerful than an attacker who is restricted to similar-data attacks. While Islam et al.~\cite{IKK} and Cash et al.~\cite{cash} motivated their attacks as similar-data attacks, decent results were only achieved in the setting of known-data attacks.

Cash et al. (CGPR) \cite{cash} defined four security levels: L1 to L4. The most secure type of scheme is referred to as L1, which only leaks the access pattern for the keywords which have been queried. The other types successively leak more until L4, which leaks the number of occurrences of each keyword and the pattern of their locations in the documents.
Islam et al. \cite{IKK} (referred to as IKK) proposed the first passive attack exploiting the access pattern leakage. After \cite{IKK}, several passive attacks have been proposed to recover the queries of L1-schemes \cite{cash, ning2018passive, blackstonerevisiting, wang2019practical}.
Although most of these attacks can be executed as similar-data attacks, they are only effective as known-data attacks with exact knowledge of at least 20\% of all indexed documents.

Other attacks have been proposed, for example, by Blackstone et al.~\cite{blackstonerevisiting}, that conceptually only work as known-data attacks and do not support similar-data attacks at all.
However, this conceptual restriction enables the attack to work with less partial knowledge and can be effective with exact knowledge down to 10\% of all indexed documents.
While this represents an impressive improvement in known-data attacks, it still requires exact knowledge of parts of the indexed documents.
On the other hand, Oya and Kerschbaum \cite{oya2021} proposed a new attack augmenting the adversary knowledge with the query frequency.
While effective, this new attack is not directly comparable to our setting with decreased attacker knowledge.
Finally, all existing attacks \cite{IKK, cash, pouliot2016shadow, oya2021, blackstonerevisiting} assume the exact knowledge of the client's keyword universe, i.e., the queryable vocabulary.

In conclusion, no practical similar-data attack has been proposed so far that achieves an accuracy higher than 50\% even under advantageous conditions (e.g., the client's keyword universe is small AND known by the attacker).
Appendix \ref{app:related_work} describes more extensively the related papers and their respective contributions.
Also, we discuss orthogonal lines of research focusing on other attacker models or schemes with different leakage profiles.

\paragraph{Our contribution.}
In this paper, we describe an attack that works \emph{without} knowledge of the indexed documents and only uses similar data.
At the same time, our similar-data attack achieves recovery rates of up to 90\%.
The documents known by the adversary only need to be distributionally close to the indexed documents.
For example, an attacker can mount a successful attack exploiting information from a previous data breach, even if the breached documents have been identified as such and were purged from the index to mitigate future known-data attacks.
For a successful similar-data attack, an adversary correctly recovers most of the queries given knowledge of only 10 query tokens and their corresponding keywords.

Our attack is based on a confidence metric used to score a trapdoor-keyword pair. The score should be maximized when the trapdoor (i.e., a query token) is paired with its (correct) underlying keyword. This confidence score is the key element that provides a good interpretability of the attack results.
We start with our \emph{score attack} that computes a confidence score of each trapdoor-keyword pair and returns, for each trapdoor, the pair with the highest score. This base similar-data attack reaches a recovery rate of 60\% while assuming around 25\% of known queries. Secondly, our paper proposes an improvement strategy reducing drastically the amount of adversary knowledge necessary, especially regarding the known queries.
Our refined score attack, an iterative refinement strategy\footnote{Code: \url{https://github.com/MarcT0K/Refined-score-atk-SSE}}, reaches recovery rates of up to 85\% with only 10 known queries in our experiments over different datasets of varying size.
More specifically, the iterative scoring approach recovers further queries by adding previously recovered queries that had a high confidence score to the set of known queries. Our attack has a low runtime and can be performed in less than two minutes.
As indicated by our experiments, the refined score attack is sensitive to the amount of knowledge available; that is, its accuracy improves with additional adversary knowledge.
This observation was not made for attacks like IKK and CGPR where their performance stays almost the same with growing amounts of known queries.
We show that both padding and obfuscation countermeasures can successfully mitigate our attack. However, these countermeasures induce storage and communication overheads. We also study how the refined score attack behaves when the attacker owns a dataset with a lower degree of similarity.

Our paper highlights that Searchable Symmetric Encryption (SSE) schemes should no longer be used without countermeasures. For example, suppose a company uses SSE to manage the employee mailboxes, each with their own encrypted index and secret key. An attacker with access to just one employee's mailbox may have enough background knowledge to successfully recover the queries of every other employee for which she has only very few known queries. Using the compromised mailbox, she can run a massive file injection attack by sending a few emails to everyone. This preliminary active attack would be a way to obtain the known queries necessary to perform the refined score attack on the rest of the employees' queries. Such a scenario is not possible with the existing known-data attacks because, by definition, they can only recover the queries from the owner of the mailbox accessible by the attacker.

\section{Definitions, attacker models, and assumptions}
The notation as introduced in this section and used throughout this work is summarized in Table \ref{tab:notations}.

\begin{table*}
	\small
	\renewcommand{\arraystretch}{1.3}
	\centering
	\caption{Summary of notations}
	\begin{tabular}{|c|c|c|}
		\hline
		Notation                                 & Meaning                                                                                                         & Size notation                          \\
		\hline \hline
		\multicolumn{3}{|c|}{\textbf{Base adversary knowledge}}                                                                                                                                             \\
		\hline
		$\mathcal{D}_{\text{sim}}$               & Similar document set                                                                                            & $n_{\text{sim}}$                       \\ \hline
		$\mathcal{Q}$                            & Queries observed by the attacker (i.e., a list of trapdoors)                                                    & $l$                                    \\ \hline
		$R_{\mathcal{Q}}$                        & Results of the queries from $\mathcal{Q}$ (i.e., a list of document identifiers for each $td$ in $\mathcal{Q}$) & $l$                                    \\ \hline
		$\text{Known}\mathcal{Q}$                & (trapdoor, keyword) pairs known by the attacker                                                                 & $k$                                    \\

		\hline \hline
		\multicolumn{3}{|c|}{\textbf{Derived adversary knowledge}}                                                                                                                                          \\ \hline
		$\mathcal{K}_{\text{sim}}$               & Vocabulary extracted from $\mathcal{D}_{\text{sim}}$                                                            & $m_{\text{sim}}$                       \\ \hline
		$C_{\text{kw}}$                          & Word-word co-occurrence matrix built from $\mathcal{D}_{\text{sim}}$                                            & $m_{\text{sim}} \times m_{\text{sim}}$ \\ \hline
		$C_{\text{td}}$                          & Trapdoor-trapdoor co-occurrence matrix built from $R_{\mathcal{Q}}$                                             & $l \times l$                           \\ \hline
		$\hat{n}_{\text{real}}$                  & Estimation of $n_{\text{real}}$, the number of documents indexed by the server                                  & not applicable                         \\

		\hline \hline
		\multicolumn{3}{|c|}{\textbf{Unknown by the attacker}}                                                                                                                                              \\ \hline
		$\mathcal{D}_{\text{real}}$              & Encrypted documents indexed by the server                                                                       & $n_{\text{real}}$                      \\ \hline
		$\mathcal{K}_{\text{real}}$              & Queryable vocabulary (i.e., the client's keyword universe)*                                                     & $m_{\text{real}}$                      \\ \hline
		$\mathcal{K}_{\text{real}}(\mathcal{Q})$ & Underlying keywords of the observed queries $\mathcal{Q}$ (i.e., the objective of the attack)*                  & $l$                                    \\ \hline
		\multicolumn{3}{c}{\raggedleft *Actually, the attacker knows a small part of this vocabulary thanks to the known queries}
	\end{tabular}
	\label{tab:notations}
\end{table*}

\subsection{Searchable symmetric encryption (SSE)}
\label{subsec:sse}
From a high-level perspective, the majority of searchable symmetric encryption (SSE) schemes are based on the same design idea introduced by Curtmola et al.~\cite{Curtmola}.
We consider a \emph{document set} $\mathcal{D}$, which consists of documents $d\in\mathcal{D}$ with identifiers $id(d)$. Each document $d$ consists of keywords. If a keyword $x$ occurs in $d$, we denote this as $x\in d$. Now, initially, a client generates an inverted index for \emph{document set} $\mathcal{D}$ that indicates for each keyword in which document it occurs. The document set $\mathcal{D}$ is encrypted on the client side using a secret key and uploaded to a server.
In a second step, the client can then query the encrypted index for single keywords using a trapdoor function, which takes the secret key and a keyword as input and outputs a unique \textit{trapdoor}. We denote as $\text{Trapdoor}(x)$ the trapdoor of the keyword $x$.
When the client searches for a keyword, she sends the corresponding trapdoor to the server.
The server computes the result set using the encrypted index together with the received trapdoor and sends back the matching result set, which consists of the matching encrypted documents and their identifiers.

Here, SSE supports various kinds of document sets such as, for example, a set of emails, a set of articles, and a set of information sheets. The only condition on the document set is that the user can extract keywords. For text files, the procedure is straightforward, but it could also be a tag extraction for videos. In the case of videos or images, the tags would be the subject of the queries. Thus, we can also consider indexing non-textual data.

Depending on the leakage profile of the scheme, the response leaks more or less information to the server. Our attack works on the minimum leakage profile called L1. L1 schemes only unveil the identifiers of the documents containing the keyword queried by the user.

\subsection{Attacker models}
A passive attacker observes the trapdoors sent to the server and the server response, which includes the list of the matching document identifiers.
These identifiers reveal no further information about the content of the document.
The attacker can link a query to its response and create (Trapdoor, DocIDs) pairs. We consider two slightly different attacker models; both are applicable for our attack as described in Section~\ref{sec:improved}:
\begin{itemize}
	\item An \textit{honest-but-curious} server follows the protocols but tries to recover the underlying keywords of the queries. To facilitate search on encrypted documents, the encrypted index is usually supposed to be stored on the server along with the encrypted documents as is the case in settings considered by the IKK and CGPR attacks. Such an attacker owns metadata about the encrypted documents: the total number of documents and their size.
	\item A \textit{passive traffic observer} records the traffic of the database. This adversary only has pairs of $(\text{Trapdoor}, \text{DocIDs})$ and uses them to recover the underlying keywords. It could also represent a case where the index server does not store the encrypted documents. Such an index server ignores the number and the size of the indexed encrypted documents.
\end{itemize}

\subsection{Adversary knowledge}

\paragraph{Similar document set}  The adversary knowledge is focused on a similar document set $\mathcal{D}_{\text{sim}}= \{d_1, \dots ,d_{n_{\text{sim}}}\}$.
A document set $\mathcal{D}_{\text{sim}}$ is similar if it is distributionally close to the indexed documents $\mathcal{D}_{\text{real}}$. Subsection \ref{subsec:def_sim} proposes a formal definition of document set similarity.

In a company, an employee's mailbox is a document set similar to her colleague's mailboxes. As another example, leaked confidential notes are a similar document set to recover the queries on the rest of the notes. Known-data attacks can also work on leaked documents, but the server can simply remove these documents from the index to avoid these attacks. Despite this removal, our ``refined'' attack is still effective. Since it is a similar-data attack, we do not need our documents to be indexed as opposed to the known-data attacks \cite{IKK, cash, blackstonerevisiting} that assume their known documents are part of the document set indexed by the server.

A vocabulary $\mathcal{K}_{\text{sim}}$ of $m_{\text{sim}}$ keywords is extracted from $\mathcal{D}_{\text{sim}}$. An index matrix is built from this document set: $\text{ID}_{\text{sim}} [i,j] = 1$ if the $j$-th keyword is contained in the $i$-th document, $0$ otherwise. Then, the $m_{\text{sim}} \times m_{\text{sim}}$ co-occurrence matrix is $C_{kw} = \text{ID}_{\text{sim}}^\top \text{ID}_{\text{sim}} \cdot \frac{1}{n_{\text{sim}}}$, where $n_{\text{sim}} = |\mathcal{D}_{\text{sim}}|$.\footnote{$A^\top$ denotes the transpose of a matrix $A$} Note that we use relative frequency numbers rather than absolute count numbers.

\paragraph{Keyword extraction algorithm}
The (keyword) distributional knowledge on document sets depends on how keywords are extracted from the documents. The attacker uses an algorithm to extract the vocabulary from her known document set.
In the literature, all attack papers (im- or explicitly) assume that the attacker uses the same keyword extraction algorithm as the client. Whether this assumption is realistic or not has not been questioned in the literature, but we would like to stress the importance of this assumption here. We also assume the attacker and the client use the same extraction algorithm, and we leave the study of the case of different extraction algorithms for future work.

\paragraph{Observed queries}
We denote as $\mathcal{K}_{\text{real}}$ the client's query-able vocabulary of $m_{\text{real}}$ keywords. The adversary does know neither this vocabulary nor its size. The adversary observes $l$ unique trapdoors and obtains their corresponding search results. Let $\mathcal{Q} = \langle td_1,  \dots , td_l \rangle$ be the set of observed trapdoors and $R_{td} = \{id(d)| (x \in \mathcal{K}_{\text{real}})\wedge(td=\text{Trapdoor}(x))\wedge(d \in \mathcal{D}_{\text{real}})\wedge(x \in d)\}$, the document identifiers returned for the trapdoor $td$.

Let $\text{DocIDs} = \{id_1,  \dots , id_p\} = \bigcup_{td \in \mathcal{Q}} R_{td}$. We note that $p \le n_{\text{real}}$.
Let $\text{ID}_{\text{real}} $ be the $p \times l$ index matrix built from the responses to the trapdoors as follows: $\text{ID}_{\text{real}} [i,j] = 1$ if the response to the $j$-th trapdoor contains the $i$-th identifier, $0$ otherwise. Finally, we can infer the $m_{\text{real}} \times m_{\text{real}}$ trapdoor co-occurrence matrix: $C_{td} = \text{ID}_{\text{real}}^\top \text{ID}_{\text{real}} \cdot \frac{1}{n_{\text{real}}}$. The estimation of $n_{\text{real}}$ is presented in Appendix \ref{estimate_n}.

\paragraph{Known queries}
As in IKK and CGPR, the adversary knows the underlying keywords of $k$ queries in $\mathcal{Q}$. The set of known queries is defined as follows:

\begin{equation*}
	\begin{split}
		\text{Known}\mathcal{Q} = \{&\langle kw_{\text{known}}, td_{\text{known}}\rangle | (kw_{\text{known}}\in \mathcal{K}_{\text{real}}\cap\mathcal{K}_{\text{sim}}) \\ & \wedge (td_{\text{known}} \in \mathcal{Q}) \wedge (td_{\text{known}}=\text{Trapdoor}({kw_{\text{known}}}))\}
	\end{split}
\end{equation*}

\subsection{Similarity definitions}
\label{subsec:def_sim}
\paragraph{Similar document set}

Let $\mathcal{C}(\mathcal{D}, kw_a, kw_b)$ be the function returning the number of co-occurrences of keywords $kw_a$ and $kw_b$ inside the document set $\mathcal{D}$. Let $\text{SimMat}(\mathcal{D}_1,  \mathcal{D}_2, \mathcal{K})$ be a function returning an $m\times m$ similarity matrix of $\mathcal{D}_1$ and $\mathcal{D}_2$ over the vocabulary $\mathcal{K} = \{kw_1, \dots kw_i, \dots kw_m\}$. The function $\text{SimMat}$ is defined such that:
\begin{equation}
	(\text{SimMat}(\mathcal{D}_1,  \mathcal{D}_2, \mathcal{K}))_{ij} = \frac{\mathcal{C}(\mathcal{D}_1, kw_i, kw_j)}{|\mathcal{D}_1|} - \frac{\mathcal{C}(\mathcal{D}_2, kw_i, kw_j)}{|\mathcal{D}_2|}
\end{equation}

In other words, the $ij$-th element of the matrix returned by $\text{SimMat}$ is the difference between the co-frequency of the keywords $i$ and $j$ in the document set $\mathcal{D}_1$ and the co-frequency of the same two keywords in the document set $\mathcal{D}_2$. Thus, $\text{SimMat}(\mathcal{D}_1,  \mathcal{D}_2, \mathcal{K})$  describes the similarity of $\mathcal{D}_1$ and $\mathcal{D}_2$ over the vocabulary $\mathcal{K}$ and the norm $||\text{SimMat}(\mathcal{D}_1,  \mathcal{D}_2, \mathcal{K})||$ is a measure of the similarity of $\mathcal{D}_1$ and $\mathcal{D}_2$. In this paper, we consider the Frobenius norm, but other norms can be considered as well.

We define $\mathcal{D}_{\text{sim}}$ and $\mathcal{D}_{\text{real}}$ as two {\it $\epsilon$-similar document sets} if for $\epsilon \ge 0$ the following holds:

\begin{equation}
	\label{eq:def_sim}
	||\text{SimMat}(\mathcal{D}_{\text{sim}},  \mathcal{D}_{\text{real}}, \mathcal{K}_{\text{real}})|| \le \epsilon
\end{equation}

In other words, the closer the co-frequencies between the document sets are, the more similar the document sets are. In our definition, we only need to consider the similarity over the queryable vocabulary $\mathcal{K}_{\text{real}}$ because those are the only keywords that are queried for by the client and to be recovered by the attacker.

\paragraph{Similar and queryable vocabularies}
To recover the queries, the attacker needs to have as many elements of the queryable vocabulary $\mathcal{K}_{\text{real}}$ as possible in her similar vocabulary $\mathcal{K}_{\text{sim}}$. This creates a natural upper bound for the attack accuracy:

\begin{equation}
	\label{eq:acc_upper_bound}
	\text{AttackAccuracy} \le \frac{|\mathcal{K}_{\text{real}} \cap \mathcal{K}_{\text{sim}}|}{|\mathcal{K}_{\text{real}}|}
\end{equation}

In other words, the attacker can only recover the queries for which the underlying keywords are contained in $\mathcal{K}_{\text{sim}}$.

In the experiments presented in Section \ref{sec:improved}, $\mathcal{K}_{\text{sim}}$ contains most elements of $\mathcal{K}_{\text{real}}$ since the average accuracy goes up to 95\% which means that more than 95\% of the keywords of the queryable vocabulary $\mathcal{K}_{\text{real}}$ are contained in the similar vocabulary $\mathcal{K}_{\text{sim}}$.

\paragraph{Attacker assumptions} An attacker knows neither the indexed documents $\mathcal{D}_\text{real}$ nor the vocabulary $\mathcal{K}_{\text{real}}$. Thus, she cannot calculate the exact $\epsilon$-similarity between her dataset $\mathcal{D}_\text{sim}$ and the indexed dataset $\mathcal{D}_\text{real}$. We \emph{assume} that:
\begin{enumerate}
	\item $\mathcal{D}_\text{sim}$ is $\epsilon$-similar to the indexed document set $\mathcal{D}_\text{real}$, for a sufficiently small $\epsilon$ (e.g. $\epsilon = 0.8$ as in the Figure \ref{fig:sim_res}).
	\item $\mathcal{K}_\text{sim}$ contains most elements of $\mathcal{K}_\text{real}$ (especially $\mathcal{K}_\text{real}(\mathcal{Q})$, the underlying keywords of the queries observed by the attacker).
\end{enumerate}

\section{Score attack}
On an intuitive level, our {\it score attack} uses a confidence metric that scores trapdoor-keyword pairs. This metric is called a \textit{matching score} and should be maximized when the trapdoor is paired with its correct underlying keyword. The attacker computes the matching score of every possible trapdoor-keyword pair. For each trapdoor, the trapdoor-keyword pair with the highest score is returned.

\subsection{Extracting the known query co-occurrence sub-matrices}
\label{subsec:sub-matrices}
The attacker uses her known queries (= known correct trapdoor-keyword pairs) to project the keyword and the trapdoor co-occurrence matrices to a common sub-vector space. This projection is done by only keeping the columns of the known queries and sorting the columns using the known queries such that the $i$-th column is related to the $i$-th known query. Formally the projection works as follows:

For a keyword $kw$, we denote its position in the vocabulary $\mathcal \mathcal{K}_{sim}=(kw_1,\ldots, kw_{m_{sim}})$ by $pos(kw)$, i.e., $pos(kw_i)=i$ for $kw_i\in \mathcal{K}_{sim}$.
Likewise, we denote the position of a trapdoor $td$ in the list of observed queries $\mathcal \mathcal{Q}=\langle td_1,\ldots,td_l\rangle$ by $pos(td)$.
We define the {\it projection of the trapdoor-trapdoor co-occurrence matrix $C_{td}$ onto the known queries} as the $l\times k$ matrix $C^s_{td}$ such that:
for all $i \in \{1 \dots l\}$ and all $j \in \{1 \dots k\}$ there exists a known query $q_j=(td_{\text{known}}, kw_{\text{known}}) \in \text{Known}\mathcal{Q}$ such that
\begin{equation}
	C^s_{td}[i,j] = C_{td}[i,pos(td_{\text{known}})].
\end{equation}

Likewise, we define the {\it projection of the word-word co-occurrence matrix $C_{kw}$ onto the known queries} as the $m_{sim}\times k$ matrix $C^s_{kw}$ such that:
for all $i \in \{1 \dots m_{\text{sim}}\}$ and all $j \in \{1 \dots k\}$ there exists a known query $q_j=(td_{\text{known}}, kw_{\text{known}}) \in \text{Known}\mathcal{Q}$ such that
\begin{equation}
	C^s_{kw}[i,j] = C_{kw}[i, pos(kw_{\text{known}})].
\end{equation}

In our notation, we use the superscript $s$ to emphasize that $C^s_{td}$  and $C^s_{kw}$ define co-occurrence \textit{sub}-matrices.
Such matrices are very convenient since we can directly compare a keyword and a trapdoor. We denote as $C^s_{kw}[kw_i]$ (resp. $C^s_{td}[td_j]$), the vector composed of the co-occurrences of keyword $kw_i$ (resp. trapdoor $td_j$) with every keyword (resp. trapdoor) related to a known query. Thus, we can extract a $k$-dimensional vector describing each keyword or trapdoor. In the next section, we will define our confidence score based on the distance between a keyword vector and a trapdoor vector.

\subsection{Confidence score and matching process}
The sub-matrices $C^s_{kw}$ and $C^s_{td}$ are used to score a trapdoor-keyword pair. A score should be maximized when the pair is correct. The {\it scoring function} for a vector-norm $\|\cdot \|$ (e.g. the L2 norm) is defined as
\begin{equation}\label{eq:score}
	\text{Score}(td_j,kw_i) = -\ln(||C^s_{kw}[kw_i] - C^s_{td}[td_j]||),
\end{equation}
for all $kw_i \in \mathcal{K}_{\text{sim}}$ and all $td_j \in \mathcal{Q}$.

Note that Equation~\eqref{eq:score} results in a high score when the distance between a keyword and a trapdoor is small. This distance can be obtained since $C^s_{kw}[kw_i]$ and $C^s_{td}[td_j]$ share a common vector space.
The $-\ln$ function is used to transform the distance into a score to focus on the order of magnitude instead of distance values close to zero.
For example, a distance of $e^{-11}$ results in a score of $11$.
In our case, the distance is always less than 1 because we are using relative frequency matrices.
We focus on orders of magnitude for interpretability matters.
Even for an attack accuracy above 80\%, the attacker needs to identify the correct predictions.
The interpretability provided by a scoring approach is then necessary. We argue that there is a higher interpretative meaning when comparing two orders of magnitude scaled between 0 and $\sim 30$ (experimental upper bound) than comparing two small norms close to zero.
Moreover, in Section~\ref{sec:improved} and in Appendix~\ref{app:clustering}, we propose geometrical methods (focusing on the distance between the scores) to improve the results of the score attack presented in this section.

Our score attack uses the score function as follows: for each trapdoor, it goes through all keywords and returns the keyword for which the score is maximized. See Figure \ref{base_alg} for an algorithmic description.
Note that the score is a confidence score, and the attacker can sort the predictions (i.e., the trapdoor-keyword pairs returned) based on their matching score to define the most likely predictions.

The norm $||\cdot||$ can be chosen freely.
However, our experiments showed that the L2 norm maximizes the accuracy.
This norm tolerates a high difference between one of the $k$ components of the vector, i.e., when one of the co-occurrences in the attacker dataset is very far from its corresponding co-occurrence in the indexed dataset.

\begin{figure}[h]
	\begin{algorithmic}
		\REQUIRE $\mathcal{K}_{\text{sim}}, C^s_{kw}, \mathcal{Q}, C^s_{td}$
		\STATE $pred \leftarrow []$
		\FORALL{$td \in \mathcal{Q}$}
		\STATE $\text{candidates} = []$
		\FORALL{$kw \in \mathcal{K}_{\text{sim}}$}
		\STATE $s = -\ln(||C^s_{kw}[kw] - C^s_{td}[td]||)$
		\STATE append $(kw, s)$ to candidates
		\ENDFOR
		\STATE $\text{candidates} = sort(\text{candidates}, desc)$
		\STATE append $(td, candidates[0])$ to $pred$
		\ENDFOR
		\RETURN $pred$
	\end{algorithmic}
	\caption{The score attack}
	\label{base_alg}
\end{figure}

\section{Experimental results}
\subsection{Methodology}
\paragraph{Datasets}
We simulate the attacks using two publicly available datasets.
First, we use the Enron dataset \cite{enron}, which is widely used in the literature to simulate attacks. Like in IKK and CGPR, we compose our Enron document set by extracting every email contained in the folders \textit{\_sent\_mail} to obtain $30 109$ documents.
Second, we use data from the Apache mailing list archives\footnote{\url{http://mail-archives.apache.org/mod_mbox/lucene-java-user/}}. We use specifically the "java-user" mailing list from the Lucene project for the years 2002-2011. This second dataset contains $50 878$ documents. It was introduced in CGPR.

\paragraph{Keyword extraction}
The keyword extraction is exclusively done on the email content. Thus, keywords of the title or the recipients' names cannot be queried. To obtain the list of the keywords, we stem the words using the Porter Stemmer\cite{porter1980algorithm} and remove the stop words. For the Apache dataset, we systematically remove the mailing list signature proposing to unsubscribe. Otherwise, the keyword contained in this "Unsubscribe" message would be useless in the search since it appears in every email.

\paragraph{Adversary knowledge generation}
At the beginning of an experiment, the document set used (i.e., Enron or Apache) is divided randomly into two disjoint subsets.
One subset is used to generate the index, i.e., the encrypted document set $\mathcal{D}_{\text{real}}$.
The second subset is part of the adversary knowledge, i.e., the similar document set $\mathcal{D}_{\text{sim}}$.
Every experiment is done with non-overlapping document sets, that is, only as similar-data attack.
The similar vocabulary $\mathcal{K}_{\text{sim}}$ is extracted from $\mathcal{D}_{\text{sim}}$ and is given to the adversary.
The index vocabulary $\mathcal{K}_{\text{real}}$ is extracted from $\mathcal{D}_{\text{real}}$ and is not known by the adversary. The similar (resp. real) vocabulary extraction algorithm consists in extracting the $m_\text{sim}$ (resp. $m_\text{real}$) most frequent keywords of $\mathcal{D}_{\text{sim}}$ (resp. $\mathcal{D}_{\text{real}}$).
The underlying keywords of the queries are chosen uniformly at random from $\mathcal{K}_{\text{real}}$.
For each run, document and query sets are freshly chosen uniformly at random.

\paragraph{Hardware and software}
The experiments are done on a Debian server with a quad-core processor (64 bits, 2.1 GHz) and 8 GB of memory.
The algorithms are implemented using Python 3.7. Specifically, we use the NLTK \cite{loper2002nltk} for the basic natural language processing: word tokenization, stemming, and stopwords.

Our code to simulate the attack and to obtain our results is publicly available: \url{https://github.com/MarcT0K/Refined-score-atk-SSE}.

\paragraph{Result presentation}
We call \textit{correct prediction}, a query for which the algorithm has returned the corresponding underlying keyword. We evaluate the performance of our attack using the term \textit{accuracy}. The accuracy corresponds to the number of correct predictions divided by the number of unknown queries.
Our accuracy excludes known queries and is always computed over the unknown queries (i.e., $ acc= \frac{|\text{CorrectPred}(\text{Unknown}\mathcal{Q})|}{|\mathcal{Q}|-|\text{Known}\mathcal{Q}|}$). In other works such as CGPR, the term \textit{recovery rate} is also used to define this concept.
We use bar plots to present the results of our experiments. Each bar is obtained by computing the average result over 50 attack simulations. These bars are completed with error bars corresponding to $\mu \pm \sigma$ with $\mu$ the average accuracy and $\sigma$ the standard deviation.

\subsection{Results}
Figure \ref{fig:base_res} shows the accuracy of the algorithm on Enron corpus for several vocabulary sizes. The server stores 60\% of the corpus and the adversary knows the remaining 40\%. The adversary has observed 15\% of the possible queries.
She knows either 15, 30, or 60 queries. When the vocabulary size is 1000 and the adversary knows 30 queries (20\% of the queries observed), the average accuracy is 60\%. When the vocabulary size is 2000 and the adversary knows 60 queries (20\% of the queries observed), the average accuracy is 55\%.
Then, the base algorithm can be successful only if it has enough known queries. When the vocabulary is larger, the accuracy decreases if the number of known queries remains the same.
The accuracy increases in function of the number of known queries (Figure \ref{fig:base_res}). We argue that we could obtain better results for large vocabularies with more known queries. Obtaining so many queries is unrealistic since we would need a preliminary inference attack or a massive injection attack to obtain the knowledge required by this base similar-data attack. Therefore, the base score attack is only a practical attack when the size of the server vocabulary is below 2000.

In Figure \ref{fig:base_res}, we can distinguish one surprising result when there are 60 known queries, and the vocabulary size is 500. This result is decreased compared to the previous one (i.e., 30 known queries), and the error bars are overlapping. We can explain that because, in this experiment, there are 75 queries for 60 known queries, i.e., only 15 unknown queries. It is the only experiment where there is a minority of unknown queries. We consider this result insignificant but keep it in our work for the complete discussion. Indeed, it does not make sense to know most of the queries before the attack and to consider the result obtained on recovering such a small amount of unknown queries.

\begin{figure}
	\includegraphics[width=\linewidth]{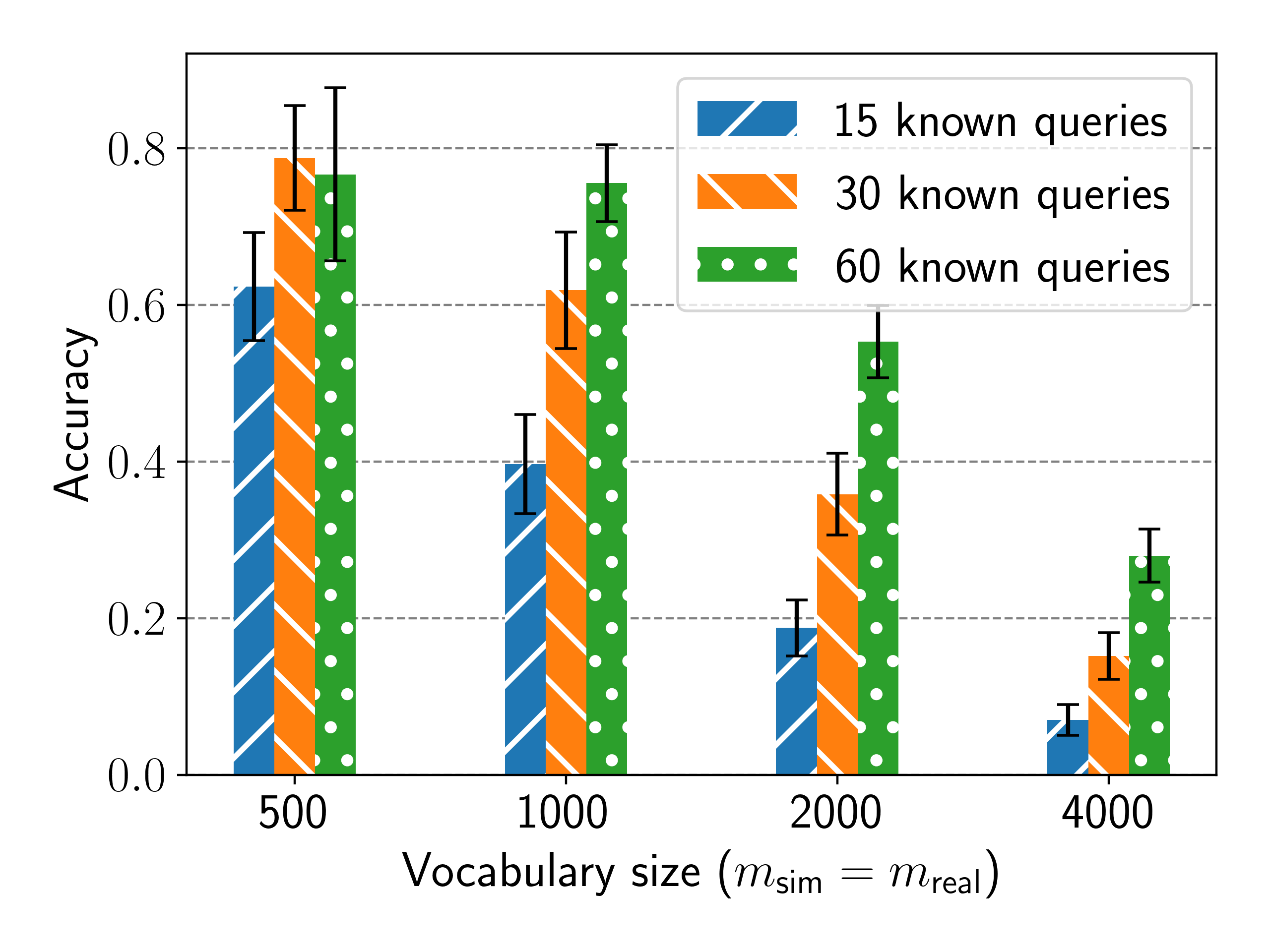}
	\caption{Average accuracy of the base matching algorithm on Enron for vocabulary size. Parameters: $|\mathcal{D}_{\text{sim}}|=12K, |\mathcal{D}_{\text{real}}|=18K, |\mathcal{Q}|=0.15\cdot m_{\text{real}}$}
	\label{fig:base_res}
\end{figure}

In \cite{cash}, it was assumed that the co-occurrences were too noisy to rely totally on them. The occurrence is much less noisy, but most of the keywords have an extremely close frequency. Therefore, it is impossible to identify a keyword just based on a single occurrence estimator except if we are sure to have perfect knowledge as in CGPR attack when they know nearly 100\% of the encrypted documents. The co-occurrence is noisier, but its distribution is scattered enough to identify keyword-trapdoor pairs. The lack of precision is balanced by the number of co-occurrences available to perform the identification. There is a trade-off between the number of estimators and the precision.

\subsection{Execution time}The complexity in time of the algorithm in Figure \ref{base_alg} is $\mathcal{O}(|\mathcal{Q}|\cdot m_{\text{sim}}\cdot k)$, if we consider the complexity of the norm $||\cdot||$ to be $\mathcal{O}(k)$. Figure \ref{fig:exec_time} describes the average execution time of this algorithm over 50 repetitions in function of the vocabulary size. We exclude the keyword extraction and the co-occurrence computation from this execution time. This experiment was done with Enron dataset. The similar document set represents 40\% of the total dataset (12 044 emails), and the server document set 60\%. Even with large document sets, the execution time is negligible (20 seconds) compared to attacks like \cite{pouliot2016shadow}, which needs 16 hours when $m_{\text{sim}}=m_{\text{real}}=1K$. Our implementation is already CPU-parallelized but can be further improved using GPU parallelization.

\begin{figure}
	\includegraphics[width=\linewidth]{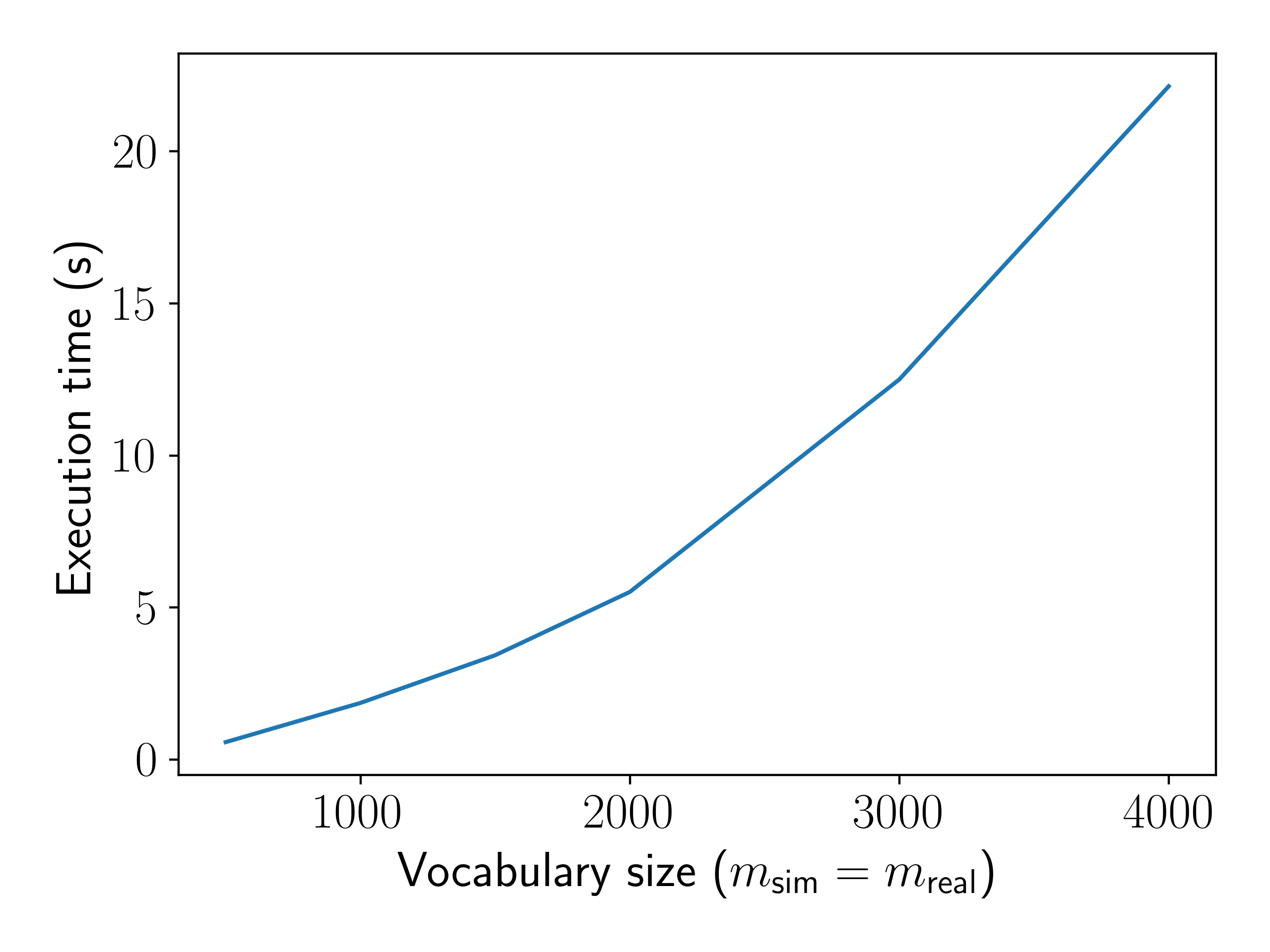}
	\caption{Average execution time of the matching algorithm for vocabulary size Parameters: $|\mathcal{D}_{\text{sim}}|=12K, |\mathcal{D}_{\text{real}}|=18K, |\mathcal{Q}|=0.15\cdot m_{\text{real}}, |\text{Known}\mathcal{Q}|=10$}
	\label{fig:exec_time}
\end{figure}

Besides its short runtime, this algorithm is deterministic and parameter-less.
Non-determinism is present, for example, in the simulated annealing used by IKK. Indeed, two runs of IKK algorithm could result in two different results because of a random choice present in this algorithm. It becomes a problem when the attack is too long to be repeated many times with different initializations and/or when the attacker does not have a confidence metric to identify the correct predictions (as in IKK).
The CGPR attack introduced an error-rate parameter, which needs to be set experimentally. However, it is unclear whether this parameter is specific to each document set or not and how to set it properly.

\section{Refined score attack}
\label{sec:improved}
\subsection{Algorithm}
Our base attack introduces a matching score that acts as a confidence metric: the higher the score is, the more likely the correctness of the keyword-trapdoor pair is.
We can use this property to determine the most certain predictions: a keyword-trapdoor candidate $(kw_i, td)$ will be considered as {\it certain} if its score is much higher than the scores of any other candidate $(kw_j, td)$.
The {\it certainty of a prediction $kw_i$ for the trapdoor $td$} is defined as:
\begin{equation}
	\text{Certainty}(td, kw_i) = \text{Score}(td, kw_i) - \max_{j \neq i} \text{Score}(td, kw_j)
\end{equation}

Based on this certainty, we propose a refinement process, which drastically reduces the number of known queries needed on the attacker's side: the matching is performed several times, and at the end of each round, the most certain predictions are added to the set of known queries. We detail this process in Figure~\ref{alg:refined_score_atk}. This algorithm introduces a new parameter: the refinement speed RefSpeed to decrease attack runtime.
However, if the refinement speed is chosen too large, it is very likely that wrong predictions are added to the known queries.
Overall, the time complexity of the algorithm in Figure \ref{alg:refined_score_atk} is $\mathcal{O}(\frac{|\mathcal{Q}|}{\text{RefSpeed}}\cdot |\mathcal{Q}|\cdot m_{\text{sim}}\cdot k)$. Notice that the runtime of the refined score attack increases by the factor $\frac{|\mathcal{Q}|}{\text{RefSpeed}}$ in comparison to the base attack, thus RefSpeed decreases the runtime by a multiplicative factor.
As a result, the refined score attack is finished in minutes, whereas Pouliot and Wright's attack \cite{pouliot2016shadow} requires several hours.

\begin{figure}
	\begin{algorithmic}
		\REQUIRE $\mathcal{K}_{\text{sim}}, C^s_{kw}, \mathcal{Q}, C^s_{td}, \text{Known}\mathcal{Q}, \text{RefSpeed}$

		\STATE $final\_pred \leftarrow []$
		\STATE $\text{unknown}\mathcal{Q} \leftarrow \mathcal{Q}$

		\WHILE{$\text{unknown}\mathcal{Q} \neq \emptyset$}
		\STATE \% 1. Extract the remaining unknown queries
		\STATE $\text{unknown}\mathcal{Q} \leftarrow \{td: (td \in \mathcal{Q}) \wedge (\nexists kw \in \mathcal{K}_{\text{sim}}: (td,kw) \in \text{Known}\mathcal{Q})\}$

		\STATE $temp\_pred \leftarrow []$

		\STATE
		\STATE \% 2. Propose a prediction for each unknown query
		\FORALL{$td \in \text{unknown}\mathcal{Q}$}
		\STATE $cand \leftarrow []$
		\COMMENT{The candidates for the trapdoor $td$}
		\FORALL{$kw \in \mathcal{K}_{\text{sim}}$}
		\STATE $s \leftarrow -\ln(||C^s_{kw}[kw] - C^s_{td}[td]||)$
		\STATE append \{"kw": $kw$, "score": $s$\} to $cand$
		\ENDFOR
		\STATE Sort $cand$ in descending order according to the score.
		\STATE $\text{certainty} \leftarrow \text{score}(cand[0]) - \text{score}(cand[1])$
		\STATE append $(td, \text{kw}(cand[0]), \text{certainty})$ to $temp\_pred$
		\ENDFOR
		\STATE
		\STATE \% 3. Either stop the algorithm or keep refining.
		\IF {$|\text{unknown}\mathcal{Q}| < \text{RefSpeed}$}
		\STATE $final\_pred \leftarrow \text{Known}\mathcal{Q} \cup temp\_pred$
		\STATE $\text{unknown}\mathcal{Q} \leftarrow \emptyset$
		\COMMENT{Stopping criteria}
		\ELSE
		\STATE Append the RefSpeed most certain predictions from $temp\_pred$ to  $\text{Known}\mathcal{Q}$
		\STATE Add the columns corresponding to the new known queries to $C^s_{kw}$ and $C^s_{td}$
		\ENDIF
		\ENDWHILE

		\RETURN $final\_pred$
	\end{algorithmic}
	\caption{The refined score attack}
	\label{alg:refined_score_atk}
\end{figure}
Each iteration of this algorithm is divided into three phases:
\begin{enumerate}
	\item Remove all (attacker-)known queries from the queries to be recovered.
	\item For each unknown query, find the best matching keyword candidate and compute its certainty score.
	\item \emph{If} there are more than RefSpeed unknown queries, we keep refining the results: the most certain predictions (RefSpeed many) are added to the known queries. Then, add the columns corresponding to the new known queries to the co-occurrence sub-matrices ($C^s_{kw}$ and $C^s_{td}$) and start a new iteration. \emph{Otherwise}, the algorithm stops and returns the queries imputed since the first iteration.
\end{enumerate}

The main benefit of the refined score algorithm is the use of more information available to the adversary. The initial score attack only uses a small part of the co-occurrence matrices (i.e., the co-occurrence sub-matrices). This refinement iteratively imputes new known queries which increases the size of these sub-matrices. We optimize the use of the adversary knowledge (the co-occurrence matrices) in order to minimize the amount needed for a successful attack (the known queries).

\subsection{Experimental results}
\label{subsec:exp_res_refined_score}
\paragraph{General comparison}
\begin{figure}
	\includegraphics[width=\linewidth]{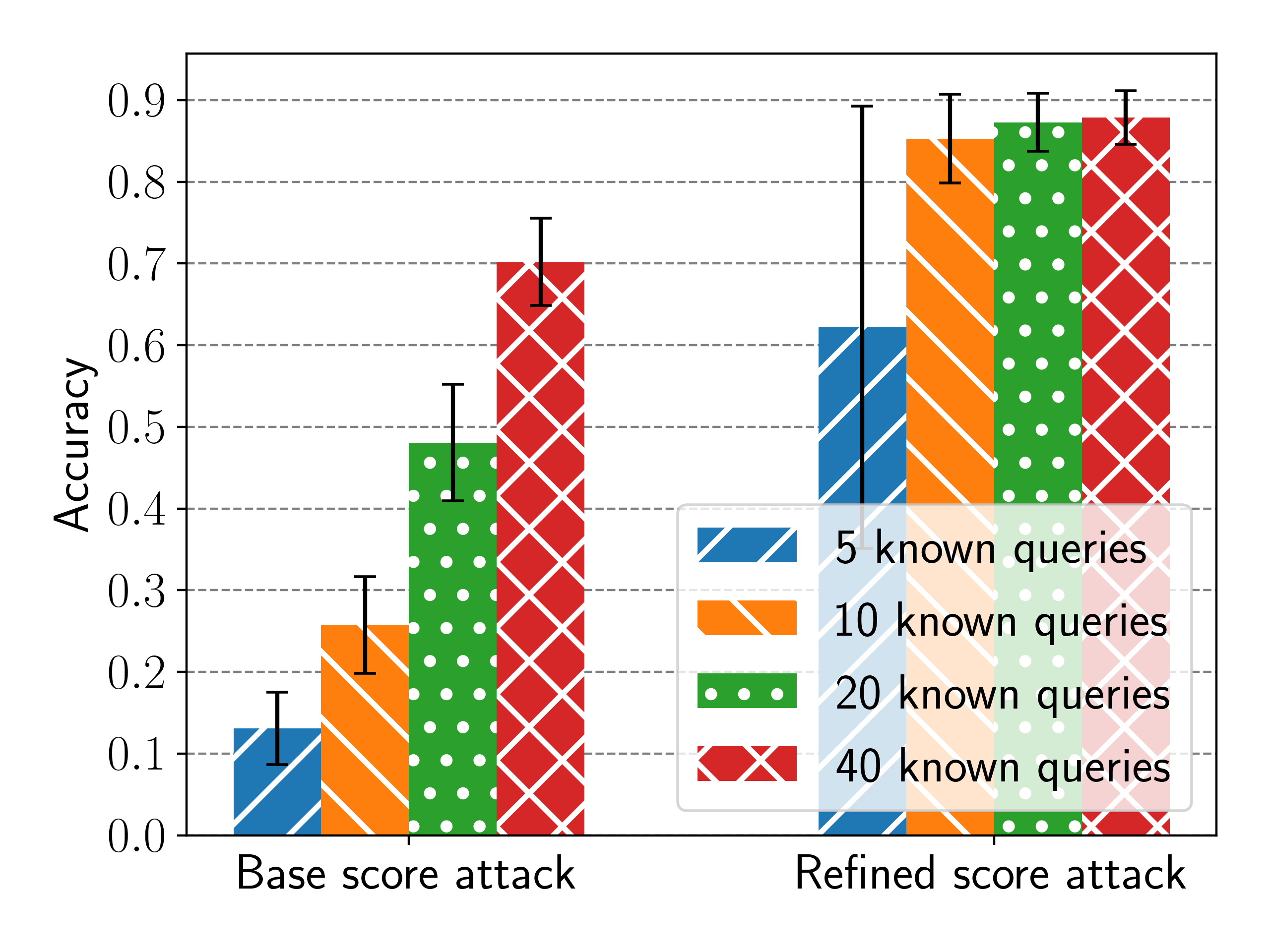}
	\caption{Comparison of the accuracy between the score attack and the refined score attack. Fixed parameters: $|\mathcal{D}_{\text{sim}}|=12K, |\mathcal{D}_{\text{real}}|=18K, m_{\text{sim}}=1.2K, m_{\text{real}}=1K, |\mathcal{Q}|=150, \text{RefSpeed}=10$}
	\label{fig:general_res}
\end{figure}
In Figure \ref{fig:general_res}, we compare the accuracy of the base and the refined versions of the attack. We fixed $m_{\text{real}} = 1K$ and $|\mathcal{Q}|=150$. Each bar represents the average accuracy over 50 simulations done with the same parameters. In Figure \ref{fig:general_res}, we see that the base algorithm needs 40 known queries to reach 70\% of accuracy while the refined score algorithm reaches 85\% with only 10 known queries. Even with only 5 known queries, the refined score algorithm achieves 62\% of accuracy with a standard deviation of 13 percentage points.

If not stated differently, we fix $|\mathcal{K}_{\text{sim}}|=|\mathcal{K}_{\text{real}}|$ to simplify the experiment understanding. However, it is likely that the similar vocabulary and the queryable vocabulary are not identical.
For the experimental results depicted in Figure \ref{fig:general_res}, we used $|\mathcal{K}_{\text{sim}}|=1.2K$ and $|\mathcal{K}_{\text{real}}|=1K$. By choosing the vocabulary sizes such that $|\mathcal{K}_{\text{sim}}| > |\mathcal{K}_{\text{real}}|$, we increase the probability that $\mathcal{K}_{\text{sim}} \cap \mathcal{K}_{\text{real}} = \mathcal{K}_{\text{real}}$.
In this case, all queries can be recovered theoretically.
In other words, as the size of $\mathcal{K}_{\text{sim}}$ increases, the accuracy upper bound as stated in Equation~\eqref{eq:acc_upper_bound} in Subsection~\ref{subsec:def_sim} potentially increases.

The refined score attack yields highly accurate results within minutes.
It recovers most of the queries and assumes less adversary knowledge than IKK and CGPR attacks.
In~\cite{cash}, Cash et al. report the average accuracy of the IKK attack is around 30\% for an attacker knowing 95\% of the indexed documents for $|\mathcal{K}_{\text{real}}|=500,|\mathcal{Q}|=150, \text{Known}\mathcal{Q}=8$.
With the same parameters, CGPR achieves 70\% accuracy.
In Figure~\ref{fig:general_res}, we see that for a vocabulary size twice as large and less known queries, i.e., $|\mathcal{K}_{\text{real}}|=1K,|\mathcal{Q}|=150, \text{Known}\mathcal{Q}=10$, the refined score attack also obtains 85\% \emph{without partial knowledge} of the encrypted documents.

\paragraph{Query set size}
\begin{figure}[t]
	\includegraphics[width=\linewidth]{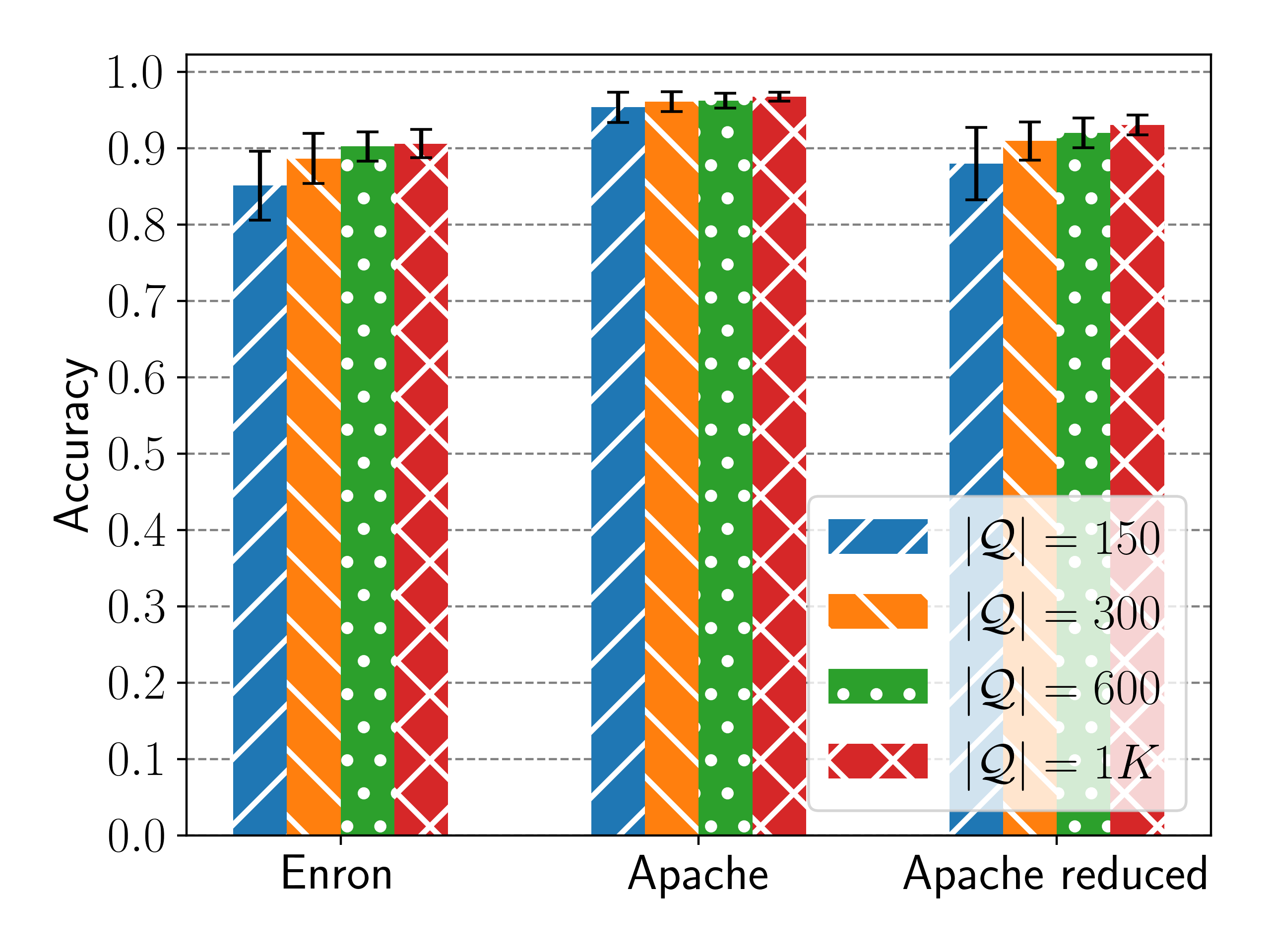}
	\caption{Comparison of the accuracy between Enron, Apache and 'Apache reduced'. Fixed parameters: $|\mathcal{D}_{\text{sim}}|=12K, |\mathcal{D}_{\text{real}}|=18K, m_{\text{sim}}=m_{\text{real}}=1K, |\text{Known}\mathcal{Q}|=15, \text{RefSpeed}=10$}
	\label{fig:enron_apache}
\end{figure}

In both the IKK and CGPR attacks, the number of observed queries was set as 15\% of all possible queries.
We investigate the role of this choice in Figure~\ref{fig:enron_apache}.
With a fixed number of known queries, the attack accuracy increases with a larger query set.
Intuitively, this demonstrates that our attack efficiently uses adversary knowledge, i.e., more observed queries imply more adversary knowledge.
In contrast, both IKK and CGPR had steady results even for an increasing number of known queries, indicating some inefficient use of adversary knowledge.

\paragraph{Different email corpus}
We compare the accuracy on Enron and on Apache in Figure \ref{fig:enron_apache}. For these simulations, we used three document sets: Enron ($|D_{\text{sim}}|=12K, |D_{\text{real}}|=18K$), Apache ($|D_{\text{sim}}|=20K, |D_{\text{real}}|=30K$) and 'Apache reduced' ($|D_{\text{sim}}|=12K, |D_{\text{real}}|=18K$). 'Apache reduced' is the Apache dataset truncated in order to have as many emails as in Enron. Apache has slightly better results than Enron while the emails contain a richer vocabulary and longer emails. Thus, our attack could be effective on a wide range of documents. Moreover, the bar plot shows that 'Apache reduced' has lower results than Apache. Our results on 'Apache reduced' are closer to those on Enron. Since Apache and 'Apacha reduced' share a common distribution and only differ in size, it shows (once again) that our attack is sensitive to the amount of adversary knowledge. In this case, the part of the adversary knowledge which is increased is the similar document set.

\paragraph{Document set similarity}
Recall the similarity definition for the document sets from Subsection~\ref{subsec:def_sim}.
For a better understanding of this new definition, we performed two experiments:

\begin{enumerate}
	\item Using Enron as an attacker document set and Apache as an indexed document set
	\item Fixing the size of the indexed document set and attacking it with similar document sets of varying size.
\end{enumerate}

During the first experiment, over 50 repetitions, we recovered at best 5 queries. This bad performance is explained by the fact that the Enron dataset has a low similarity with the Apache dataset ($\epsilon=10.2$). Further, the attacker vocabulary ($\mathcal{K}_\text{sim}$) and the queryable vocabulary ($\mathcal{K}_\text{real}$) only have 56\% of their keywords in common.
In comparison, Figure~\ref{fig:query_distrib} shows results for experiments where the attacker and the server share up to 98\% of their vocabulary.
Recall from Subsection~\ref{subsec:def_sim} that the joint keywords are an upper bound for the attack accuracy.
The disjoint vocabulary set combined with the high $\epsilon$ value between Enron and Apache explains the low attack accuracy for the first experiment.
Enron is composed of emails sent in a company, while Apache is composed of emails from a mailing list dedicated to a highly technical project. This important difference in the nature of the emails results in two very different keyword distributions (i.e., a very low similarity between the document sets).

\begin{figure}
	\includegraphics[width=\linewidth]{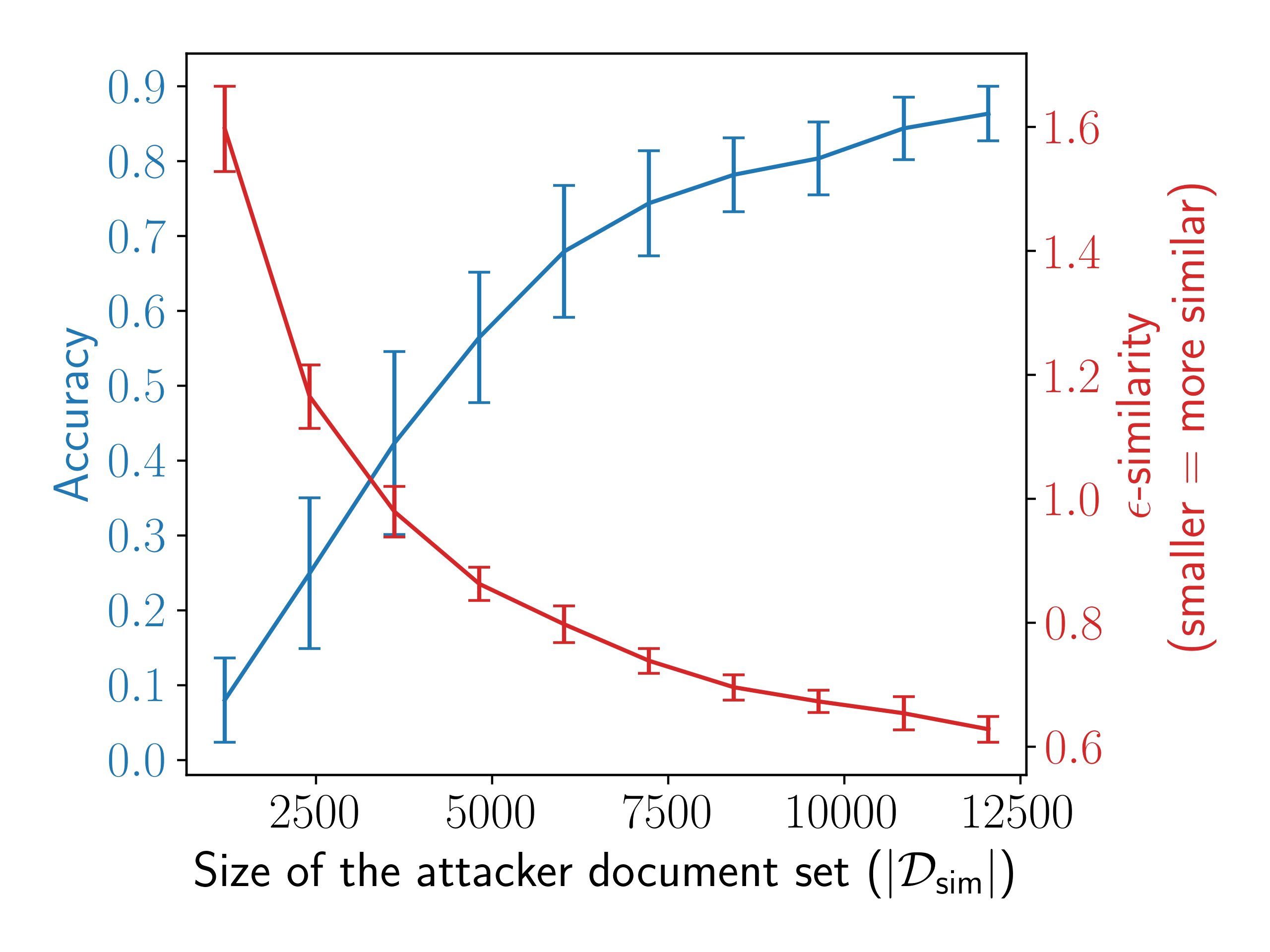}
	\caption{Accuracy and $\epsilon$-value of set similarity for varying attacker document set sizes with Enron. Fixed parameters: $|\mathcal{D}_{\text{real}}|=18K, m_{\text{sim}}=m_{\text{real}}=1K, |\mathcal{Q}|=150, |\text{Known}\mathcal{Q}|=15$}
	\label{fig:sim_res}
\end{figure}

We show the results of the second experiment in Figure~\ref{fig:sim_res}. By varying the size of the attacker dataset, the co-occurrence matrices of the attacker dataset and the indexed dataset diverge more or less.
In other words, this size reduction applies noise to the attacker's word-word co-occurrence matrix compared to the one computed with the complete dataset.
We preferred to apply this size reduction instead of applying a synthetic gaussian noise (as is done in the IKK attack paper) to the matrix because the added noise is more realistic this way.
Figure~\ref{fig:sim_res} shows that reducing the attacker document set size leads to increased $\epsilon$ values
Hence, Ii is an efficient way to decrease similarity.

In Figure~\ref{fig:sim_res}, we observe that the smaller the attacker dataset is, the less similar the document sets are. Hence, the less accurate the attack results are.
When the attacker dataset size is divided by 2, e.g., from 12K documents to 6K, we still achieve an average accuracy of 68\%.
Thus, the refined score attack shows a certain degree of robustness against decreased similarity. However, if we further reduce the size of the dataset, the accuracy is also further reduced until we have an ineffective attack.

\subsection{Attack analysis}
\label{subsec:attack_analysis}
\paragraph{Role of the amount of information}
The refined score attack is sensitive to the amount of information given to the attacker. The more information the adversary has, the higher the attack accuracy is. This holds true for each piece of information owned by the attacker: document set, observed query set, and known queries.
This was not the case in the previous attacks, especially for IKK and GCPR, which had identical results with and without known queries. IKK presented an accuracy of 80\% regardless of the percentage of known queries (from 0 to 25\% in their article).
CGPR only presented results without known queries for their count attack even if it could use them.

\paragraph{Technical comparison with related attacks}
Technically, all query-recovery attacks solve a matching problem between trapdoors and keywords based on specific background information available to the attacker.
IKK assumes partial knowledge of the indexed documents together with known trapdoor-keyword mappings. IKK describes the matching problem as an optimization problem that minimizes the distance between the trapdoor co-occurrence matrix and the keyword co-occurrence matrix.
CGPR makes similar assumptions while, in practice, it does not require known trapdoor-keyword mappings. CGPR iteratively filters keyword-trapdoor candidates for which the differences between the occurrences (computed from attacker documents and from observed queries) do not match.
Blackstone et al.~\cite{blackstonerevisiting} and Oya and Kerschbaum~\cite{oya2021} both propose attacks using query volume information only. \cite{blackstonerevisiting}~assumes an attacker can identify known documents in the index and thus still requires partial knowledge of indexed documents. They represent two bipartite graphs: one connecting indexed documents to the queries and one connecting known documents to keywords; then, they match query nodes with keyword nodes by iteratively refining the candidates using multiple filtering steps. Oya and Kerschbaum~\cite{oya2021} formulate an optimization problem based on maximum likelihood estimators, which assumes a distributional knowledge of the indexed documents plus knowledge about query frequency.

Instead of partial index information, we focus on a few attacker-known keyword-trapdoor pairs, which we use to score every keyword-trapdoor candidate. We then iteratively add pairs with the highest scores to the attacker-known pairs to improve our knowledge and refine further predictions. This avoids complex optimization problems and the requirement of knowledge about indexed documents. The scoring and its corresponding iterative refinement are the core novelties of our attack.

As highlighted in~\cite{blackstonerevisiting}, all prior attacks require exact knowledge of the queryable vocabulary (i.e., the client's keyword universe). Our attacker does not require such knowledge and builds her own vocabulary. Considering the following setup: $|\mathcal{D}_{\text{sim}}|=12K, |\mathcal{D}_{\text{real}}|=18K, m_{\text{sim}}=m_{\text{real}}=1K, |\mathcal{Q}|=300, |\text{Known}\mathcal{Q}|=15$, with an exact knowledge of the queryable vocabulary, we obtain an average accuracy of 92\%. On the other hand, when the attacker builds her own knowledge, we obtain an average accuracy of 87\%. This accuracy decrease is a direct consequence of the accuracy upper bound presented in Equation~\eqref{eq:acc_upper_bound} of Subsection~\ref{subsec:def_sim}.

In Appendix \ref{app:substitution}, we detail the relation between substitution cipher cryptanalysis and SSE attacks (especially the refined score attack).

\paragraph{Improving the attack using clustering}
Our novel scoring approach offers further possibilities for improvement. In Appendix \ref{app:clustering} we discuss clustering to improve attack results further.
In our attacks, when a prediction is uncertain, we sometimes have a group of candidates with higher scores than the rest of the candidates instead of only one candidate with a particularly high score.
In such cases, returning a list of potential keywords seems natural instead of forcing the algorithm to choose only one keyword.
Note that it would not affect the overall interpretability of the results as the scores are augmented.
In Appendix~\ref{app:clustering}, we show that clustering can further increase the accuracy of the refined scoring attack.

\section{Attack mitigation}
\subsection{Existing countermeasures}
To mitigate leakage-abuse attacks, several countermeasures have been proposed in \cite{IKK, cash, chen2018dprivacy, xu2019hardening}. We divide these countermeasures into two categories: padding (formalized by CGPR) and obfuscation. IKK proposed a first countermeasure that could be assimilated into padding. Padding consists in adding fake entries, i.e., fake keyword-document pairs. The user can easily filter these false-positive results when they receive the database response. With padding, there is no entry removal because it could impact the search results (i.e., no false negative results). To harden this countermeasure, Xu et al. proposed in \cite{xu2019hardening} a method to produce fake entries that cannot be distinguished from the real entries by an attacker. In \cite{chen2018dprivacy}, Chen et al. presented a new kind of countermeasure: obfuscation. First, it uses code erasure to divide the documents into shards. Thanks to code erasure, false negative results are allowed because the user does not need every shard to reconstruct the document. After having computed the shards, the algorithm adds and removes shards from the results. The removal rate is chosen so the reconstruction rate for matching documents is close to 100\%. Thus, false-negative shards do not result in false-negative documents.

Chen et al. also presented an improved attack scenario where the attacker knows which shards belong together. In this case, the countermeasure corresponds to padding because the attacker knows that all the reconstructed documents are either matching document or a false-positive result. Moreover, he knows that the proportion of matching documents which is not reconstructed is negligible. Therefore, if the attacker only keeps the reconstructed files, he would have all the matching documents plus some false-positive results.

These countermeasures have been proposed to mitigate known-data attacks, but they are also suitable for similar-data attacks since they alter the co-occurrence matrix $C_{td}$ inferred from the queries. Therefore, padding and obfuscation should be also effective in mitigating our attack.

\subsection{Experimental results}
To test the possibility of mitigating our attack, we implemented the padding presented in CGPR and the obfuscation presented in \cite{chen2018dprivacy}. For padding, we use the countermeasure proposed by CGPR, which is well established but the hardening proposed by Xu et al. \cite{xu2019hardening} would not provide highly different results since we do not try to filter fake entries.
Figure \ref{fig:countermeasure} describes the average accuracy of the refined score attack over 50 simulations for several vocabulary sizes. For the padding, we used a padding size $n_{\text{pad}}=500$. For the obfuscation, we used the parameters used by Chen against the "improved" attack: $p=0.88703$ the rate of false-positive shards, $q=0.04416$ the rate of false-negative shards.

\begin{figure}
	\includegraphics[width=\linewidth]{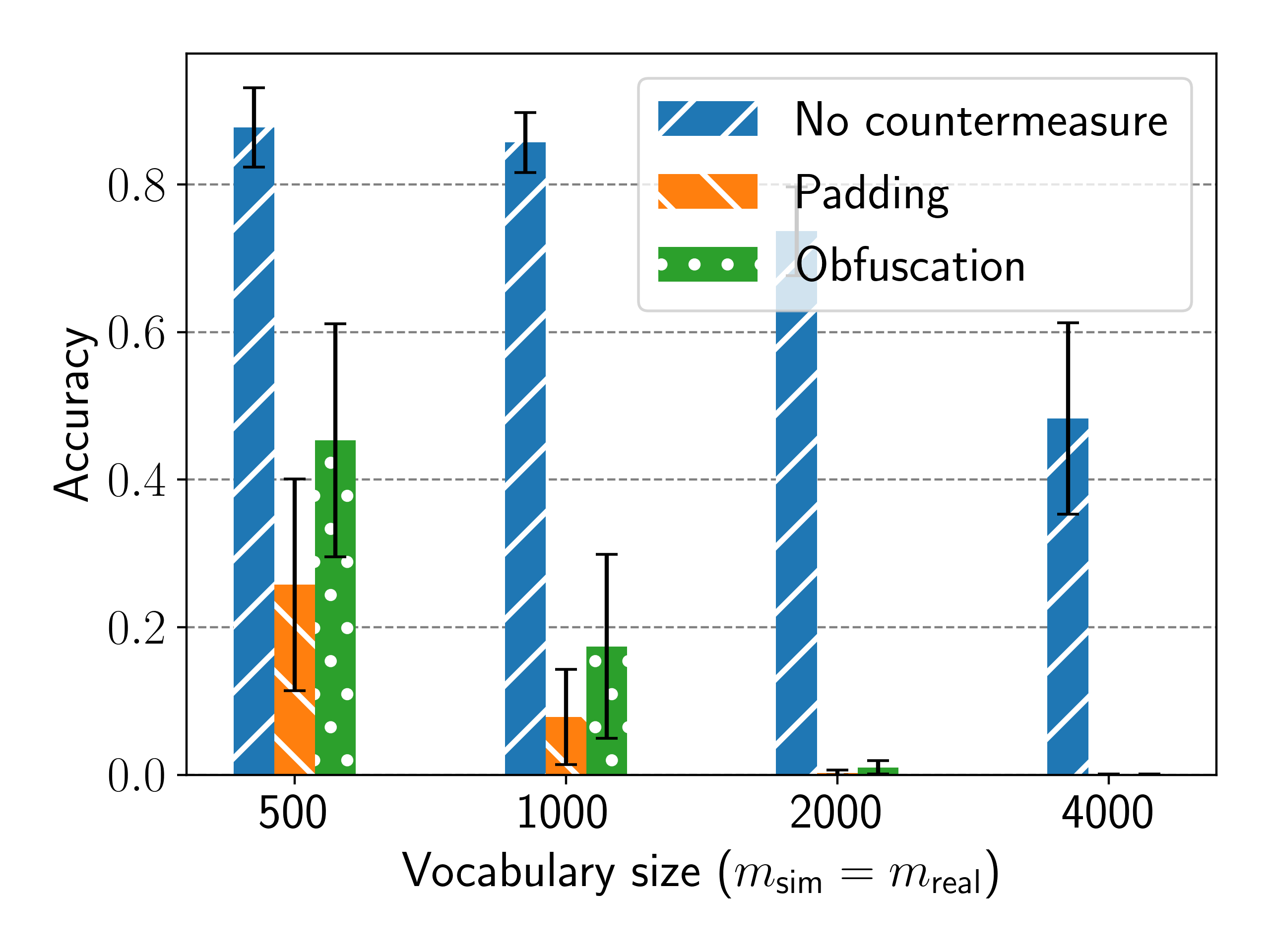}
	\caption{Comparison of the accuracy for countermeasures. Fixed parameters: $|\mathcal{D}_{\text{sim}}|=12K, |\mathcal{D}_{\text{real}}|=18K, |\mathcal{Q}|=0.15\cdot m_{\text{real}}, |\text{Known}\mathcal{Q}|=15$. Padding: $n_{\text{pad}}=500$. Obfuscation: $p=0.88703, q=0.04416$}
	\label{fig:countermeasure}
\end{figure}

Figure \ref{fig:countermeasure} clearly shows good mitigation from both countermeasures. For small vocabularies, the accuracy can still be considered as too high. However, as the vocabulary size grows, the accuracy becomes small and negligible for big vocabularies. This figure should not be used to compare the efficiency of the countermeasures. Padding performs better than obfuscation because the padding size we chose is high. For example, when $|\mathcal{K}_{\text{real}}|=1K$, the number of entries is increased by 32\% because of padding. When $|\mathcal{K}_{\text{real}}|=4K$, the number of entries is increased by 166\% because of padding. These fake entries create several types of overheads, including storage, communication and computation. Chen et al. chose $p=0.88703$ and $q=0.04416$ to minimize the overheads then, it is likely that obfuscation can achieve results equivalent or better with bigger overheads. We leave the comparison of these countermeasures and their overheads for future work. Our experiments highlight the importance of hiding the document access pattern to mitigate the refined score attack.

Our attack provides a matching score which can help to identify good predictions.  When $|\mathcal{K}_{\text{real}}|=500$, the average accuracy for padding is 35\% and for obfuscation 47\%, the refined score attack can successfully identify a non-negligible part of the correct predictions thanks to the matching score.
It is needed to define a maximum query recovery rate, so the countermeasure parameters are chosen so that there is no attack with an accuracy higher than this threshold. An analogy with encryption security is possible: the attack accuracy is the adversary advantage and the maximum query recovery is the threshold under which the advantage is considered negligible. We leave this direction for future work.

\section{Additional results}
\subsection{Generalization}
\label{subsec:generalization}
In \cite{bost2017}, Bost and Fouque presented the word-word co-occurrence as an order 2 of co-occurrence, occurrence being the order 1 of co-occurrence. Thus, we generalize our attack and build n-dimensional co-occurrence tensors to work on co-occurrences of order $n$. This generalization help to recover the queries since it increases exponentially the number of co-occurrence we can rely on.
For example, let us consider order 3: a word-word-word co-occurrence. We build 3-dimensional co-occurrence tensors and $C_{kw}[i,j,k]$ (resp. $C_{td}[i,j,k]$) is the number of documents where the keywords (resp. trapdoors) $i$,$j$ and $k$ appear together.

Our attack remains identical and just the matrix construction differs. In the refined score attack, if the order $n>2$, each keyword and trapdoor is represented by a $(n-1)$-dimensional tensors and the matching score will be computed via a matrix norm. The main issue of this generalization is the space complexity $\mathcal{O}((m_{\text{sim}})^n)$ due to tensor sizes. For $m_{\text{sim}}=1K$, with order 2, the similar co-occurrence matrix has 1 million cells and with order 3, the similar co-occurrence tensor has 1 billion cells. The first reason which could justify not to use an order greater than 2 is the technical limitations.

\begin{table}
	\small
	\renewcommand{\arraystretch}{1.3}
	\caption{Accuracy statistics on Enron over 50 simulations of orders 2 and 3. $|\mathcal{D}_{\text{sim}}|=12K,|\mathcal{D}_{\text{real}}|=18K,m_{\text{real}}=m_{\text{sim}}=300, |\mathcal{Q}|=75, |\text{Known}\mathcal{Q}|=10$}
	\centering
	\begin{tabular}{||c|c|c||}
		\hline
		Accuracy statistics & $\mu$  & $\sigma$ \\ \hline
		Order-2 attack      & 0.92   & 0.0351   \\ \hline
		Order-3 attack      & 0.77
		                    & 0.0659
		\\ \hline
	\end{tabular}
	\label{tab:generalization}
\end{table}

We have done simulations to compare order 2 and order 3. For each order, we run 50 simulations with Enron dataset, $m_{\text{real}}=m_{\text{sim}}=300, |\mathcal{Q}|=75, |\text{Known}\mathcal{Q}|=15$. As shown in Table \ref{tab:generalization}, for order 2, we obtained an average accuracy of 92\% and for order 3, 77\%. Then, in our case, increasing the order decreased the accuracy. It highlights the trade-off between the number of co-occurrence estimators and the noise of these co-occurrences. To make a decision, we need a maximum of co-occurrence estimators, but if they are too noisy, they will be misleading and the decision may be wrong. Here, we only have 30 thousand emails to compute 1 billion co-occurrences which are not enough to limit the noise of the co-occurrence tensor. However, increasing the co-occurrence order may be a viable option for attacks on larger datasets.

Note that the real co-occurrence matrix $C_{td}$ is always built using the index matrix ($ID[i,j]=1$ if document $i$ contains the underlying keyword of query $j$), whatever the order is. Thus, we expect altering the index matrix as proposed by IKK to be an effective countermeasure even against generalized attacks.

\subsection{About the observed query distribution}
\begin{figure}
	\includegraphics[width=\linewidth]{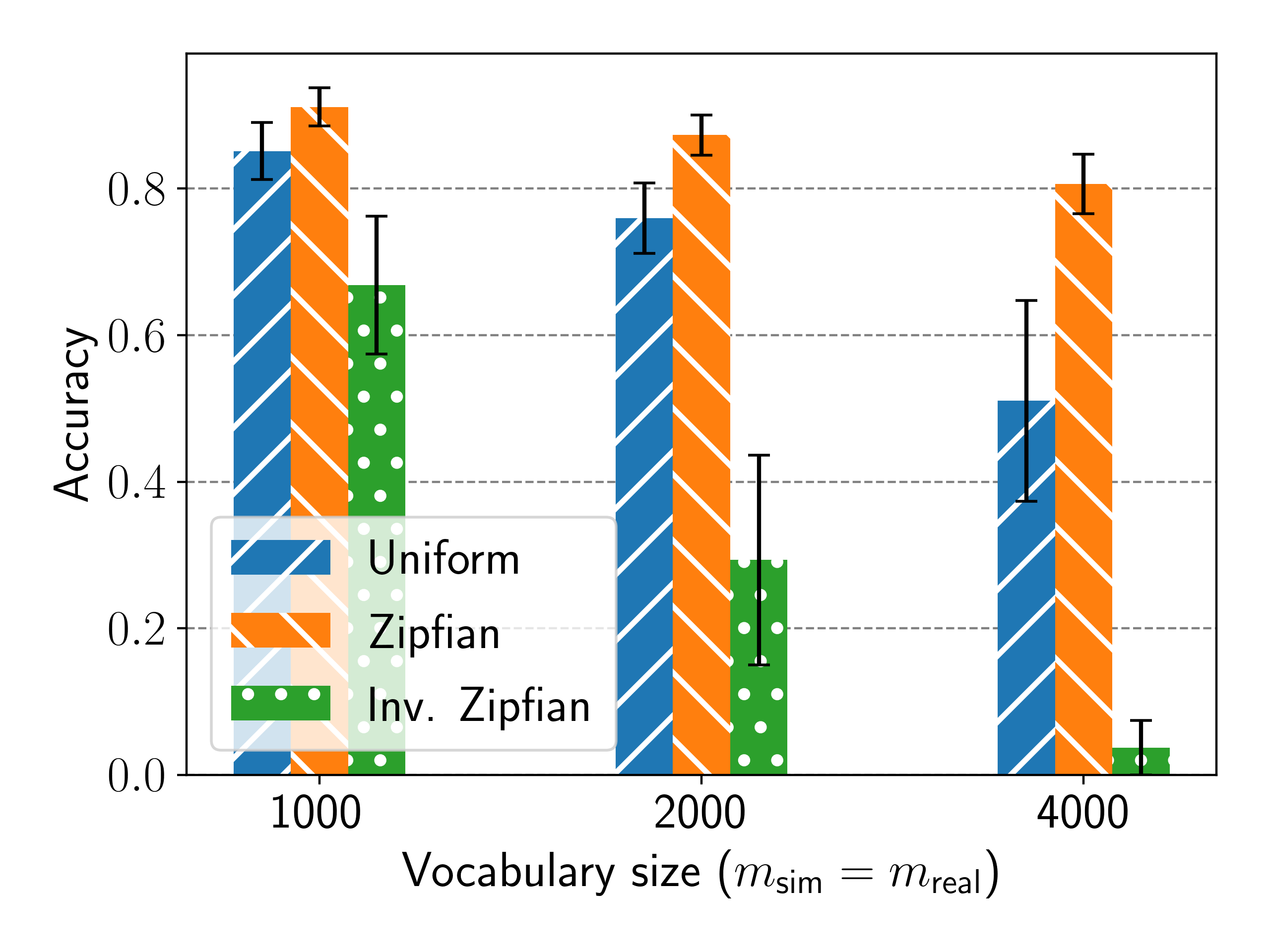}
	\caption{Comparison of the accuracy of the refined score attack for different query distributions. Fixed parameters: $|\mathcal{D}_{\text{sim}}|=12K, |\mathcal{D}_{\text{real}}|=18K, |\mathcal{Q}|=0.15\cdot m_{\text{real}}, |\text{Known}\mathcal{Q}|=15$.}
	\label{fig:query_distrib}
\end{figure}

In \cite{kornaropoulos2019state}, Kornaropoulos et al. discuss the default choice of the uniform distribution for range queries.
While they focus on SSE schemes allowing range queries, the same statement can be done for single-keyword search schemes. In \cite{blackstonerevisiting}, Blackstone et al. criticized the role of the query distribution. They show that the attack performance is highly impacted whether the attack is executed on the most frequent keyword or not.

Figure \ref{fig:query_distrib} compares the accuracy of the refined score attack over three different query distributions: Uniform, Zipfian and inverted Zipfian.
These probability distributions are defined in Appendix \ref{app:prob}.
The uniform distribution is the standard setup used in the rest of our experiments. With Zipfian distribution, which gives more weight to the highest rank elements, we mostly obtain queries for which the underlying keyword is one of the most frequent. With inverted Zipfian distribution, which gives more weight to the lowest rank elements, we mostly obtain queries for which the underlying keyword is one of the least frequent.

Figure \ref{fig:query_distrib} shows that inverted Zipfian decreases the refined score attack accuracy.
The attack becomes ineffective when the vocabulary is bigger (i.e., a vocabulary size of 4000).
Despite the inverted Zipfian distribution, the refined score attack still achieves 67\% of accuracy when the vocabulary size is 1K.
On the other hand, with the Zipfian distribution, the refined score attack reaches 81\%, when the vocabulary size is 4000, and up to 91\% when the vocabulary size is 1K.

The refined score attack could be much more devastating if the uniform assumption turns out to be false. By default, the literature uses the uniform distribution for the queries.
This assumption could be dangerous because, if the real query distribution is more advantageous than the uniform distribution (e.g. the Zipfian distribution), the SSE schemes are way more exposed than what is usually admitted.
The gap between Uniform and Zipfian distributions in Figure \ref{fig:query_distrib} for a vocabulary size of 4000 is particularly alarming since it nearly doubles the attack accuracy considering Zipfian distribution instead of Uniform distribution.

\subsection{About the known query distribution}
The query distribution explains only a part of the result variance. The distribution of known queries also impacts the results. It means that some known queries provide more information than others. To identify the impact of this distribution, we simulated 50 times two refined score attacks and studied their respective accuracy variance. The first attack is the basic setup used in our article: 5 known queries picked uniformly among the queries. The second attack simulation picks 5 known queries uniformly from the 25\% of queries with the largest result sets, i.e., the most frequent underlying keyword. We report the results in the Table \ref{tab:known_query_distribution}. The basic setup has a bigger variance and a lower mean. The second experiment presents steadier results which confirms that the distribution of known queries impacts the results.
Since an attacker has still chances to observe only queries with the most frequent underlying keywords given a uniform distribution, the maximum accuracy scores are equivalent for both distributions.

\begin{table}
	\renewcommand{\arraystretch}{1.3}
	\caption{Variance of the accuracy over 50 simulations of the refined score attack. $|\mathcal{D}_{\text{sim}}|=12K,|\mathcal{D}_{\text{real}}|=18K,m_{\text{real}}=m_{\text{sim}}=1K, |\mathcal{Q}|=150, |\text{Known}\mathcal{Q}|=5$.}
	\centering
	\begin{tabular}{||c|c|c|c|c|c|c||}
		\hline
		Acc. stats.            & $\mu$ & $\sigma$ & $q_{0.25}$ & $q_{0.75}$ & min  & max  \\ \hline
		Base setup             & 0.65  & 0.21     & 0.54       & 0.78       & 0.06 & 0.87 \\ \hline
		Top 25\% $\mathcal{Q}$ & 0.71  & 0.16     & 0.68       & 0.81       & 0.17 & 0.87 \\
		\hline
	\end{tabular}
	\label{tab:known_query_distribution}
\end{table}

The variance is lower when the adversary obtains more known queries because there are enough "good" known queries to start the refinement. Thus, only a part of these known queries are truly necessary. An attacker can use \cite{zhang2016all} active attack to obtain their known queries. Thus, an attacker can aim at a specific known query distribution in order to minimize the number of known queries needed by attacking specific keywords. Just few qualitative known queries are needed to start a successful refined score attack.

\section*{Conclusion}
We introduced a highly effective similar-data attack against SSE. The refined score attack achieves an accuracy (i.e., query-recovery rate) of 90\% while only using documents \emph{similar} to the encrypted documents. Previous attacks could only achieve equivalent results by assuming that attacker knows a significant part of the encrypted documents (from 20\% for Blackstone et al. \cite{blackstonerevisiting} to 70\% for Cash et al. \cite{cash}). Our attack provides devastating results while avoiding the strong assumption that the attacker knows a part of the encrypted documents.
Unlike the existing attacks we assume few known queries (around 10 known queries in our experiments over different datasets of varying size) rather than knowing the plaintext of a substantial part of the encrypted documents.
We argue that it is more realistic to obtain few known queries (e.g. using active attacks) than to obtain a part of the documents indexed. Thus, the refined score attack is more easily performed while previous attacks had a limited number of realistic use cases.
One conclusion of our experiments is that the more knowledge the adversary has, the better our refined score attack performs. Despite the simplicity of this statement, it was not observed in previous attacks. This sensitivity to the information amount highlights an optimized adversary knowledge use as opposed to a relative underutilization of this knowledge by some existing attacks.
We also showed that the existing countermeasures (padding and obfuscation) can effectively mitigate the refined score attack.
Finally, we highlighted that the distribution of observed and known queries impacts the accuracy of our attack. It implies that, if the real query distribution is different from the uniform distribution commonly used in the literature, the refined score attack can be even more devastating.
Considering the results presented in this article, SSE should no longer be used without countermeasures.


\bibliographystyle{plain}
\bibliography{ref}

\newtheorem{dissim}{Definition}
\newtheorem{linkcrit}[dissim]{Definition}
\newtheorem{hierar}[dissim]{Definition}
\newtheorem{levels}[dissim]{Property}
\newtheorem{notations}[dissim]{Notation}

\newtheorem{ordered}[dissim]{Theorem}
\newtheorem{distances}[dissim]{Theorem}
\newtheorem{bestcandidate-cluster}[dissim]{Definition}
\newtheorem{lem}[dissim]{Lemma}

\appendix

\section{Extended discussion of related works}
\label{app:related_work}
We have discussed the most relevant related work in Section~\ref{sec:introduction} and compared our refined score attack with the related attacks in Section~\ref{subsec:attack_analysis}. In this Appendix, we would like to discuss other lines of related research.

\subsection{General overview}
\paragraph{Passive attacks against L1 schemes}
Several attacks have already been proposed to recover the queries of L1 schemes. All of them had one common assumption to achieve good results: the adversary knows at least a part of the documents indexed. In \cite{IKK} (IKK), Islam et al. presented the first attack using the co-occurrence of the search tokens. In \cite{cash} (CGPR), Cash et al. presented a simpler but very effective known-data attack using slightly more knowledge than IKK. Two significant advantages of CGPR over IKK were its execution time and its effectiveness on large keyword sets. However, \cite{cash} highlighted that IKK and CGPR attacks could not provide compelling results as similar-data attacks, i.e., the attacker needs to know a part of the indexed data to recover some queries. In \cite{pouliot2016shadow}, Pouliot and Wright introduced a graph-matching attack. Pouliot et al.~\cite{pouliot2016shadow} proposed a graph-matching attack that can work as a known-data attack and a similar-data attack. As a similar-data attack, it rarely recovers more than 50\% of the queries and only under rather advantageous conditions (e.g., small vocabulary). Moreover, the execution time increases tremendously for bigger keyword sets.
In \cite{ning2018passive}, Ning et al. presented a new known-data attack that represents every keyword (resp. search token) with a binary sequence: the $i$-th bit is 1 if the $i$-th document contains the keyword (resp. matches the query), 0 otherwise.
The sequences are compared to find the underlying keywords of the queries. This attack performs better than CGPR. However, it is a known-data-only attack since it cannot use distributional knowledge to construct these binary sequences. In \cite{blackstonerevisiting}, Blackstone et al. revisited the underlying concepts of the known-data attacks and presented a new known-data attack. This attack can be done on co-occurrence-hiding schemes and achieves good results even when the known data is small. Despite these results, this attack cannot be executed using similar data (i.e., it is a known-data-only attack) because the algorithm assumes that the documents known by the attacker are indexed. In \cite{wang2019practical}, Wang et al. introduced a volume-based attack working on co-occurrence-hiding schemes. The latter two attacks~\cite{blackstonerevisiting,wang2019practical} even work against schemes with protection beyond L1 security.

All these passive attacks tried to avoid known queries as part of the adversary knowledge: \cite{wang2019practical, blackstonerevisiting} did not assume known queries.
In contrast, \cite{IKK} and \cite{cash} showed that query knowledge does not improve the results significantly.
However, they all made another strong assumption: the adversary knows a significant part of the encrypted documents.
Our work follows an orthogonal approach and study similar-data attacks that avoid partial knowledge of the encrypted documents but exploit a small subset of queried keywords. We argue that obtaining a few known queries (using active attacks) is easier than obtaining a part of the encrypted documents. Thus, the refined score attack can be considered much more feasible than existing leakage abuse attacks.

\paragraph{Other attacks}
Other attacks exist against L1 schemes but differ from what is intended by the passive attacks presented above. In \cite{zhang2016all}, Zhang et al. presented an active attack that recovers specific keywords based on maliciously crafted files/documents injected by the attacker.
In \cite{liu2014search}, Liu et al. recover the queries exploiting the query frequencies. This last direction has not been treated much more in the literature since no query dataset is available. In \cite{oya2021}, Oya and Kerschbaum introduced an attack assuming co-occurrence AND query frequency. This new attacker model is more or less a mixture of the attacker models from Liu et al.~\cite{liu2014search} and Blackstone et al.~\cite{blackstonerevisiting}. While strengthening the attacker assumptions with the query frequency, it is important to highlight that this attack can work on schemes with no or partial access pattern leakage \cite{demertzis2020seal, patel2019mitigating, patel2020lower}.

Several articles have also presented attacks on schemes other than L1 schemes. In \cite{giraud2017practical, grubbs2017leakage, cash}, the attacks are focused on less secure schemes, i.e., schemes with L2, L3, or L4 leakage profiles. In \cite{grubbs2018pump, kellaris2016generic, lacharite2018improved}, attacks on schemes supporting range queries were proposed. These SSE schemes are opposed to the schemes for single-keyword searches that are the subject of our attack.

\subsection{Relation between substitution cipher cryptanalysis and SSE attacks}
\label{app:substitution}
There exists a particular link between SSE passive query-recovery attacks and the cryptanalysis of substitution ciphers. In \cite{cash}, Cash et al. define L4 leakage profile as \textit{full plaintext under deterministic word-substitution cipher}.
We argue that the analogy to substitution ciphers still holds for most secure schemes, including L1 schemes. In simple substitution ciphers, each plain letter is replaced by one other letter, the key is a dictionary, for example $\{"a": "x", "b": "j", \dots\}$.  In SSE with single-keyword search, we can construct a similar mapping such that $\{trapdoor_i: keyword_j, \dots\}$.
The main difference is the larger alphabet size for SSE.

To perform a ciphertext-only attack on substitution ciphers, a frequency analysis is necessary.
A very common way to proceed is to compute $n$-grams from the ciphertext and compare them with the reference $n$-grams occurrences computed from a large publicly available corpus, e.g.~the most frequent English bigram is "th".
Several methods, automated or not have been proposed: \cite{gaines1956cryptanalysis, forsyth1993automated, spillman1993use, dhavare2013efficient, clark1997parallel, uddin2006cryptanalysis}.
Especially in \cite{gaines1956cryptanalysis} and \cite{dhavare2013efficient}, simple attacks based on "digrams" are presented. The digrams are letter-letter co-occurrence matrices. We can see an equivalence between the letter-letter co-occurrences used for substitution cipher attacks and the word-word co-occurrences used by IKK, CGPR and our attack. However, we note that, with L1 schemes, the matrix is symmetric while the digrams are not. Moreover, Forsyth and Safavi-Naini in \cite{forsyth1993automated} tackle the frequency analysis problem (for substitution cipher attacks) by using simulated annealing as IKK for SSE. In \cite{dhavare2013efficient}, Dhavare et al. presented a hill climbing solution, an optimization algorithm similar to simulated annealing. Simulated annealing is a powerful approach against substitution ciphers and less against SSE due to the large alphabet size.

We argue that a similar-data attack on SSE is analogous to a ciphertext-only attack on substitution ciphers since the publicly available corpus used for substitution ciphers attack is analogous to the similar document set used in similar-data attacks. In our case, we present a similar-data attack with known queries corresponding to a chosen-plaintext attack. One could say that substitution ciphers cryptanalysis uses $n$-grams and not only bigrams. Thus, we present a generalized attack in Subsection \ref{subsec:generalization} that uses co-occurrence of order n (i.e., the number of documents containing $n$ specific keywords).

\paragraph{Refined score attack}
We compare the refined score attack to the digram method presented in \cite{gaines1956cryptanalysis}. This method has a preliminary step of vowel identification based on the letter occurrences. This preliminary step can correspond to the prior active attack performed to obtain known queries for the refined score attack. Then, the cryptanalyst identifies iteratively new letters using the digrams of these vowels (equivalent to co-occurrence matrices). When the cryptanalyst guesses new letters, she can use them to identify the remaining unknown letters. In our refined score attack, at the end of each iteration, we learn few queries and will use these newly known queries to recover the remaining unknown queries. There is a strong similarity between the notions of known letters in substitution cipher cryptanalysis and of known query in SSE attack and the way they are used to iteratively discover new letters (resp. queries).

\section{Estimation of the number of indexed documents}
\label{estimate_n}
Both IKK and GCPR attacks use known queries but conclude that the results are equivalent with or without them. We assume that known queries convey significant information and should be fully used to obtain an effective attack as shown in Section \ref{sec:improved}.
Another example of this knowledge underutilization is the number of documents indexed $n_{\text{real}}$, which is considered as known by IKK and CGPR attacks. However, if the attacker is a passive traffic observer he would not have this information. IKK and CGPR only considered the honest-but-curious server. Storing the index and the documents on two separate servers is a simple way to degrade the information leakage to that of a passive traffic observer.

This number is mandatory to transform the count matrix into a frequency matrix. We note $\mathcal{D}_{\text{sim}}(kw)$, the documents from $\mathcal{D}_{\text{sim}}$ that contains the keyword $kw$. We also highlight $|R_q| = |\mathcal{D}_{\text{real}}(kw)|$ if $kw$ is the underlying keyword of query $q$.

\begin{equation}
	\label{eq:n_real}
	\hat{n}_{\text{real}}=	\frac{1}{k} \cdot  \sum_{kw,td \in \text{Known}\mathcal{Q}} {\frac{|R_q|}{|\mathcal{D}_\text{sim}(kw)|}}  \cdot   n_{\text{sim}}
\end{equation}

Equation \eqref{eq:n_real} shows how $\hat{n}_{\text{real}}$ the estimation of the number of indexed is computed. The first part of the equation (i.e., $\frac{1}{k} \cdot  \sum_{kw,td \in \text{Known}\mathcal{Q}} {\frac{|R_q|}{|\mathcal{D}_\text{sim}(kw)|}}$) is the average ratio between the number of encrypted documents containing one keyword and the number of similar documents containing the exact same keyword. Then, this ratio (which is a sort of scale factor) is multiplied by the number of similar documents to obtain $\hat{n}_{\text{real}}$.

Thanks to this estimation, the minimum adversary knowledge needed by IKK and CGPR attacks does not include the number of indexed documents contrary to what was implicitly assumed. If the result length is hidden, the co-occurrence between the known queries can be used to estimate $\hat{n}_{\text{real}}$.

\section{Improvement strategy: Clustering}
\label{app:clustering}

The matching score provides a very interesting basis to interpret and analyse the results. By default, we always pick $td_\text{pred}=\arg \max_i \text{Score}(td_i, kw)$ and the difference between the score of $td_\text{pred}$ and the score of the second best prediction is considered as the certainty of the predictions. However, we observed that, sometimes, we have several potential candidates instead of one:
\begin{itemize}
	\item Classical score distribution: [$\dots$ 6, 6.2, 6.3, 9], one clear candidate
	\item Atypical score distribution: [$\dots$ 6, 6.2, 6.3, 7.9, 8, 8.2], one cluster of candidates
\end{itemize}

We argue that it would be very interesting to return clusters when the choice is uncertain. To process appropriately these score distributions, we use hierarchical clustering (\cite{ward1963, intro_info_retrieval}) to identify the best-candidate cluster. With clustering, the prediction will be a cluster (either with one single candidate or with several candidates) and the certainty of the prediction will be the distance between the best-candidate cluster and the rest of the scores. In the main body of this paper, a certain prediction was a prediction for which the certainty is high. In this case, a certain prediction is a single-point cluster for which the certainty is high.

Hierarchical clustering is an iterative method used to obtain $n-1$ clusters from $n$ clusters. We specifically use single-linkage clustering, which considers the minimum distance between two clusters as the dissimilarity (described in Figure \ref{alg:single-linkage-clust} of Appendix \ref{app:def}).
Usually it is needed to define a number of clusters or a "cutting height" to know when to stop the iterations. To avoid this problem, we define a maximum size $\text{MaxSize} < m_{\text{sim}}$ for our best-candidate cluster. This parameter can be easily set by an attacker without knowledge about hierarchical clustering.

For each query, we have $m_{\text{sim}}$ candidates because we compute the matching score of the trapdoor with all keywords from $\mathcal{K}_{\text{sim}}$.
We denote $\Gamma_i$, the clusters after the $i$-th iteration and $\Gamma_0$ the initial state which is a partition of one-point clusters. The clustering is done over $\mathcal{S} = \{s_1,\dots,s_{m_{\text{sim}}}\}$, the matching scores of one given trapdoor with all the candidates, sorted in descending order. We define the best-candidate cluster $S_{max}$ as follows:

\begin{align*}
	 & \exists i_{max} \in \{0,\dots, m_{\text{sim}} - 2\} \text{ such that}                           \\
	 & \quad i_{max} = max\{i: \exists S \in \Gamma_i, s_1 \in S \wedge |S| \le \text{MaxSize} \}      \\
	 & \text{So, the best-candidate cluster } S_{max} \in \Gamma_{i_{max}} \text{and } s_1 \in S_{max} \\
	 & \quad \text{i.e., } S_{max}\text{ is the cluster containing the highest score.}
\end{align*}

Since we use single-linkage clustering in a 1-dimensional space (i.e., the scores), we obtain the Equation \eqref{eq:improved_alg} (proof in Appendix \ref{app:main}).
This does not hold with more than one dimension or with a different linkage criterion. This equation is interesting because it implies that there is a $\mathcal{O}(\text{MaxSize})$ algorithm to find $S_{max}$. In comparison, the naive single-linkage clustering algorithm is $\mathcal{O}(n^3), n\gg\text{MaxSize}$. This complexity reduction is important because a clustering is performed over the $m_{\text{sim}}$ candidates of each trapdoor.

\begin{equation}
	\label{eq:improved_alg}
	\begin{aligned}
		\exists & i \le \text{MaxSize}, S_{max} = \{s_1, \dots ,s_i\} \text{ and }  \\
		        & \forall j \le \text{MaxSize} + 1, s_i - s_{i-1} \le s_j - s_{j-1}
	\end{aligned}
\end{equation}

To obtain $S_{max}$, we use Figure \ref{alg:clust}, which takes as input the score set $\mathcal{S}$ and the parameter $\text{MaxSize}$.
It outputs the best-candidate cluster and the distance between this cluster and the closest cluster (i.e., the certainty of the prediction).
To find the best-candidate cluster, the algorithm needs to find the maximum leap between two consecutive scores among the (MaxSize + 1) maximum scores from the score set $\mathcal{S}$. From Equation \eqref{eq:improved_alg}, we know that all the scores, which are before this maximum leap compose $S_{max}$.

\begin{figure}[h]
	\begin{algorithmic}
		\REQUIRE $\mathcal{S}, \text{MaxSize}$
		\STATE $\text{MaxDist} \leftarrow 0$
		\STATE $\text{MaxInd} \leftarrow 0$
		\STATE $\mathcal{S} \leftarrow sort(\mathcal{S}, desc)$

		\FORALL{$i \in {1 \dots \text{MaxSize}}$}
		\STATE $\text{CurrDist} = S[i] - S[i+1]$
		\IF{$\text{MaxDist} < \text{CurrDist}$}
		\STATE $\text{MaxDist} \leftarrow \text{CurrDist}$
		\STATE $\text{MaxInd} \leftarrow i$
		\ENDIF
		\ENDFOR

		\STATE $S_{max} = S[:\text{MaxInd}]$
		\COMMENT {MaxInd first elements of S}
		\RETURN $S_{max}, \text{MaxDist}$
	\end{algorithmic}
	\caption{Best-candidate clustering algorithm}
	\label{alg:clust}
\end{figure}

This clustering can be used to improve either the base attack or the refined attack. To improve the base attack, we just need to call the clustering algorithm in the prediction loop: instead of appending the candidate with the highest score, the algorithm appends the best-candidate cluster to the prediction list. To improve the refined attack, clustering will be used to identify the most certain predictions. The algorithm stops when there are less than RefSpeed single-point clusters found. We present comparative results in Figure \ref{fig:comp_res_clust}. The accuracy is strongly increased for the base score attack (about 15 percentage points). We still observe an improvement for the refined score attack (about 5 percentage points).

\begin{figure}
	\includegraphics[width=0.9\linewidth]{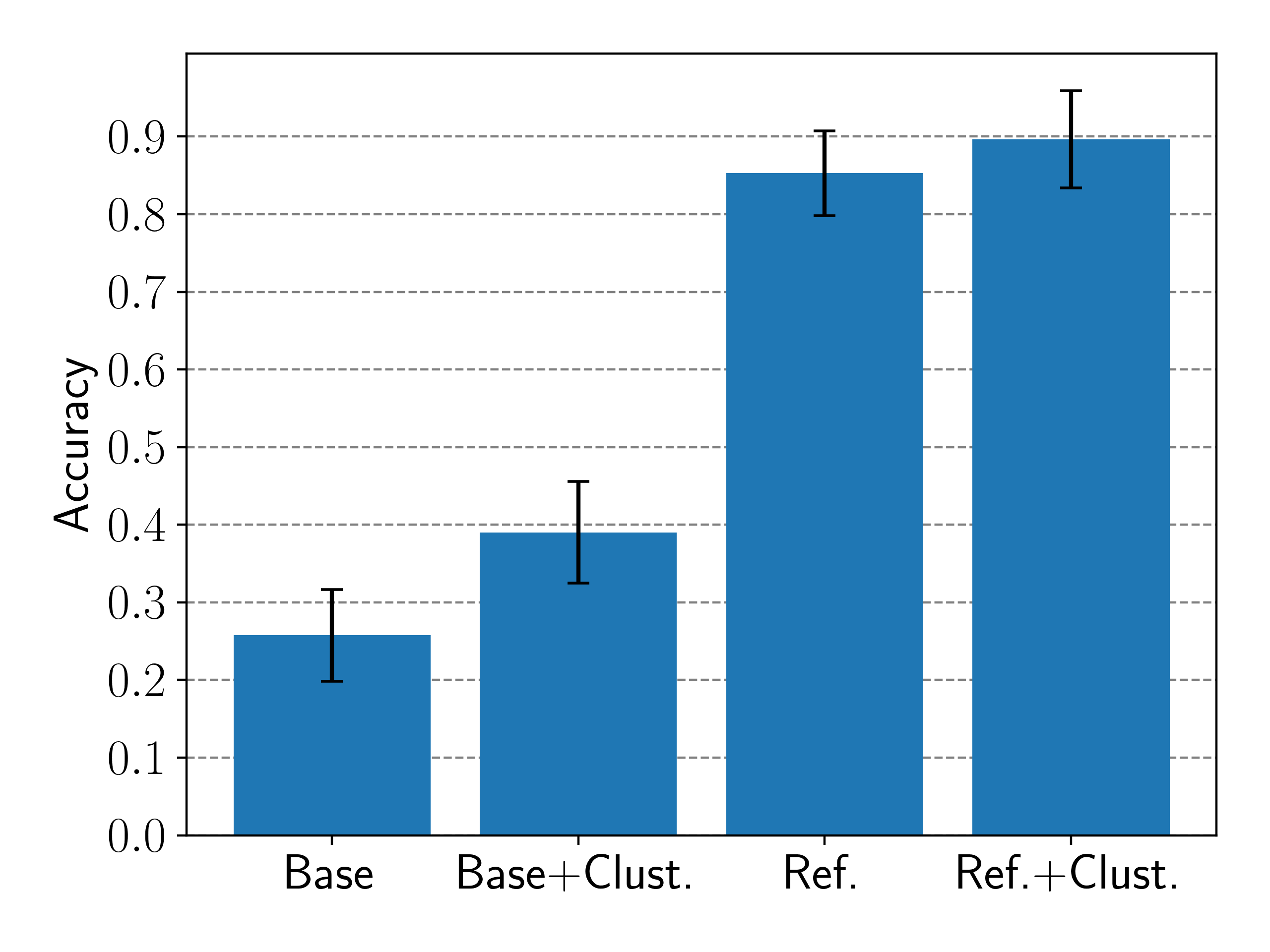}
	\caption{Comparison of the accuracy of the score attacks with and without clustering. Parameters: $|\mathcal{D}_{\text{sim}}|=12K, |\mathcal{D}_{\text{real}}|=18K, m_{\text{sim}}=1.2K, m_{\text{real}}=1K, |\mathcal{Q}|=150, |\text{Known}\mathcal{Q}|=10, \text{RefSpeed}=10, \text{MaxClustSize}=10$}
	\label{fig:comp_res_clust}
\end{figure}

In the particular case of clustering, a correct prediction is a prediction for which the cluster returned contains the correct keyword. Thus, comparing the accuracies with and without clustering is imperfect since methods with clustering has a slightly different definition of accuracy. Moreover, we highlight that, by construction, the methods improved  with clustering must perform at least as well as the standard methods.

\paragraph{Cluster size choice}
\begin{table}
	\small
	\renewcommand{\arraystretch}{1.3}
	\caption{Cluster size statistics (and their corresponding average accuracy) over 50 simulations of the refined score attack using the clustering improvement. $|\mathcal{D}_{\text{sim}}|=12K,|\mathcal{D}_{\text{real}}|=18K,m_{\text{real}}=m_{\text{sim}}=1K, |\mathcal{Q}|=150, |\text{Known}\mathcal{Q}|=15$}
	\centering
	\begin{tabular}{||c|c|c|c|c|c||c||}
		\hline
		Size stats. & $\mu$ & $q_{0.8}$ & $q_{0.85}$ & $q_{0.95}$ & $q_{0.99}$ & Acc.                \\ \hline
		MaxSize=1   & 1     & 1         & 1          & 1          & 1          & 0.873               \\ \hline
		MaxSize=5   & 1.26  & 1         & 1          & 3          & 5          & 0.902               \\ \hline
		MaxSize=10  & 1.36  & 1         & 2          & 3          & 7          & 0.906               \\ \hline
		MaxSize=20  & 1.41  & 1         & 2          & 4          & 8          & 0.907               \\ \hline
		MaxSize=50  & 1.45  & 1         & 2          & 4          & 9          & 0.907               \\ \hline
		\hline
		\multicolumn{7}{||c||}{Results below were obtained with $m_{\text{real}}=m_{\text{sim}}=2K$} \\ \hline
		MaxSize=5   & 1.35  & 1         & 2          & 3          & 5          & 0.658               \\ \hline
		MaxSize=10  & 1.53  & 2         & 2          & 5          & 8          & 0.667               \\ \hline
		MaxSize=20  & 1.63  & 2         & 2          & 5          & 11         & 0.670               \\ \hline
	\end{tabular}
	\label{tab:cluster_size}
\end{table}

Table \ref{tab:cluster_size} presents the size statistics of the clusters returned by the clustering + refined score attack algorithm with varying $\text{MaxClustSize}$. The table is separated into two parts: the upper part presents results when the vocabulary size is 1K and the lower part when the vocabulary size is 2K. First, we note that choosing MaxSize=1 is strictly equivalent to using the standard refined score attack. In the upper part, we read that $q_{0.8}=1$, which means that for at least 80\% of the queries, only one possible keyword is returned. When $\text{MaxClustSize}=10$, we also note that $q_{0.99}=7$, i.e., less that 1\% of the queries has a best-candidate cluster reaching the maximum size. These results tend to prove that the clustering does not improve artificially the results because the refined score algorithm returns cluster only for a small minority of results. Moreover, when $\text{MaxClustSize}=1$ (i.e., refined score attack without clustering), the accuracy is slightly decreased (3\%).

In the Figure \ref{fig:comp_res_clust}, we use $\text{MaxClustSize}=10$. In the upper part of Table \ref{tab:cluster_size}, we show that the accuracy is increased compared to the experiments using 1 or 5 as the maximum size. However, the accuracy is only very slightly (less than 0.1\%) increased when the maximum size is 20 or 50. This small accuracy difference could also be few big clusters (i.e., 20-point clusters) containing the correct keyword but the attacker has no way to identify this result as a correct prediction. We can also wonder how this attacker can exploit such clusters. Thus, the experimental accuracy might be increased but the practical accuracy would remain identical. On the other hand, choosing a maximum cluster size of 10 instead of 20 divides the complexity by two. To sum up, by choosing $\text{MaxClustSize}=10$, we sacrifice an uncertain 0.1\% accuracy gain for an algorithm execution time divided by two.

In our experiments, these clusters seem to contain words, which are semantically close. We observe clusters containing only figures or only days of the week. However, we cannot draw any strong semantic conclusion from these clusters since they are built from very small corpuses. Clusters with a real semantic signification are used in natural language processing, especially for translation but are obtained from corpuses composed of billions of documents. This claim seems coherent since the word-word co-occurrence matrix is the basis of word embeddings as GloVe \cite{pennington2014glove}.

\subsection{Definitions}
\label{app:def}

\begin{linkcrit}
	The \emph{linkage criterion} determines the distance between sets of vectors $A$ and $B$. We denote as $d$ the Euclidean distance. The most common linkage criteria are:
	\begin{itemize}
		\item $D_{min} = min\{d(a,b): a\in A, b\in B\}$
		\item $D_{max} = max\{d(a,b): a\in A, b\in B\}$
	\end{itemize}
\end{linkcrit}

\begin{hierar}
	Hierarchical clustering is an iterative method using a linkage criterion in order to obtain $n-1$ clusters from $n$ clusters. When the linkage criterion is $D_{min}$, it is called single-linkage clustering.
\end{hierar}

\begin{figure}[h!]
	\begin{algorithmic}
		\REQUIRE $\mathcal{S} = \{s_1\dots s_n\}$

		\STATE $\Gamma \leftarrow \text{EmptyList}(size=n+1)$
		\STATE \# The list containing the consecutive clustering.
		\STATE $\mathcal{L} \leftarrow \text{EmptyList}(size=n+1)$
		\STATE \# The list containing the consecutive levels of clustering

		\STATE $\Gamma[0] \leftarrow \{\{s_1\}\dots \{s_n\}\}$
		\STATE $\mathcal{L}[0] = 0$

		\FOR{$i=1$ \TO $n$}
		\STATE $S_1,S_2 = arg\,min \{D_{min}(S,S'): S,S' \in \Gamma[i-1], S \neq S'\}$
		\STATE $\Gamma_{temp} \leftarrow \Gamma[i-1]$
		\STATE remove $S_1, S_2$ from $\Gamma_{temp}$
		\STATE append $S_1 \cup S_2$ to $\Gamma_{temp}$
		\STATE $\Gamma[i] \leftarrow \Gamma_{temp}$
		\STATE \# $\Gamma[i] = (\Gamma[i] \backslash \{S_1,S_2\}) \cup \{S_1\cup S_2\}$
		\STATE $\mathcal{L}[i] = D(S_1,S_2)$
		\ENDFOR
	\end{algorithmic}
	\caption{Naive single-linkage clustering algorithm}
	\label{alg:single-linkage-clust}
\end{figure}

Note that the function $\mathcal{L}$, called the level, is defined in Figure \ref{alg:single-linkage-clust}. Each clustering $\Gamma_{i}$ is linked to its level $\mathcal{L}(i)$. The function $\mathcal{L}$ satisfies the following equation:
\begin{equation}
	\label{eq:prop_lvl}
	\forall i, \mathcal{L}(i+1) \ge \mathcal{L}(i)
\end{equation}

\begin{notations}
	\label{notations}
	Let $\mathcal{S}=\{s_1 \dots s_{n}\}$ be the set of one-dimensional vectors we want to cluster. We assume that $\mathcal{S}$ is sorted in descending order (so $s_j \geq s_{j+1}$). Note that if $s_j=s_{j+1}$, their index can be swapped. \\
	Let $\Gamma_i$ be the clustering of $\mathcal{S}$ obtained after $i$ iteration of the for-loop of Figure \ref{alg:single-linkage-clust}.
\end{notations}

\subsection{Auxiliary theorems}
\label{app:side_th}
\begin{ordered}{The clusters from $\Gamma_i$, as defined in Notation \ref{notations}, are ordered, i.e.,}
	\label{th:order}
	\begin{align*}
		\forall i \in & \{0 \dots n-1\}, \forall S\in \Gamma_i, \exists j \in \{1 \dots |\mathcal{S}|\} \text{ such that } \\
		              & S=\{s_j\dots s_{j+|S|-1}\}
	\end{align*}

\end{ordered}
\begin{proof}
	We proceed by induction.

	For $i=0$, it is true because $\Gamma_0$ only contains one-point clusters.

	Let us suppose the statement is true for $i=p$ with $p<n-1$, and we will prove that:
	\begin{equation}
		\forall S \in \Gamma_{p+1}, \exists s_{j} \in \mathcal{S} \text{ such that } S=\{s_j \dots s_{j+|S|-1}\}
		\label{eq:ordered_p+1}
	\end{equation}

	\begin{multline*}
		\exists S_a, S_b \in \Gamma_p \text{ such that } \forall S,S' \in \Gamma_p, \\
		D_{min}(S_a, S_b) \le D_{min}(S, S') \\
		\shoveleft{\implies \Gamma_{p+1} = (\Gamma_p \backslash \{S_a, S_b\}) \cup \{S_a \cup S_b\}}
	\end{multline*}

	We assume without loss of generality $max(S_a) \le min(S_b)$.
	By definition, the statement is true for $\Gamma_p$. We observe that $\Gamma_{p+1} \backslash \{S_a \cup S_b\} \subset \Gamma_{p}$ so it is true for $\Gamma_{p+1} \backslash \{S_a \cup S_b\}$. Then, we only need to prove that:
	$$
		\exists j \in \{1 \dots n\} \text{ such that }, S_a\cup S_b = \{s_j \dots s_{j+|S_a\cup S_b|-1}\}
	$$

	We note $ind$ the function that returns the index of an element: $ind(s_j)=j$. We denote $s_{j^1}, s_{j^2} \in \mathcal{S}$ such that $S_a=\{s_{j^1} \dots \}, S_b=\{s_{j^2} \dots \}$.

	First, let us suppose $\nexists s_{j'} \in \mathcal{S} \text{ such that } ind(min(S_a)) < j' < ind(max(S_b))$, it implies that $S_a\cup S_b = \{s_{j^1} \dots s_{j^1+|S_a|-1}, s_{j^2} \dots s_{j^2 + |S_b| -1}\}$. In this case, Equation \eqref{eq:ordered_p+1} is true.

	Second, we suppose $\exists s_{j'} \in \mathcal{S} \text{ such that } ind(min(S_a)) < j' < ind(max(S_b))$.  First, we note that:
	\begin{multline*}
		ind(min(S_a)) < j' < ind(max(S_b)) \\
		\implies s_{j'} \in [min(S_a), max(S_b)].
	\end{multline*}

	By induction hypothesis, we know that $s_{j'} \notin ]min(S_a), max(S_a)[$ and $s_{j'} \notin ]min(S_b), max(S_b)[$. Thus, there are 3 possibilities:
					\begin{enumerate}
						\item $s_{j'} = min(S_a)$
						\item $s_{j'} = max(S_b)$
						\item $s_{j'} \in [max(S_a), min(S_b)]$
					\end{enumerate}

					Case $s_{j'} = min(S_a)$, we just need to swap the index of  $s_{j'}$ and $min(S_a)$ which is possible because they are equal. Then, Equation \eqref{eq:ordered_p+1} becomes true.

					Case $s_{j'} = max(S_a)$, in the same way as in the previous case, Equation \eqref{eq:ordered_p+1} becomes true.

					Case $s_{j'} \in [max(S_a), min(S_b)]$: we can first highlight that $s_{j'} \notin ]max(S_a), min(S_b)[$. Otherwise, we would have:
	\begin{multline*}
		D_{min}(\{s_{j'}\}, S_a\cup S_b) < D_{min}(S_a, S_b) \\
		\implies \mathcal{L}(\Gamma_{p+1}) < \mathcal{L}(\Gamma_p) \qquad \qquad
	\end{multline*}

	This last implication is impossible because of Equation \eqref{eq:prop_lvl}. Thus, either $s_{j'} = max(S_a)$ or $s_{j'}=min(S_b)$.

	If $s_{j'} = max(S_a)$ then $min(S_a)=max(S_a)$. Otherwise, it would exist $i'<p+1$ such that $\mathcal{L}(\Gamma_{p+1})=0<\mathcal{L}(\Gamma_{i'})$, which is impossible. Since we have $min(S_a)=s_{j'}$, we can switch their index so Equation \eqref{eq:ordered_p+1} is true.
	If $s_{j'} = min(S_b)$, same reasoning.

	Finally, we always have $S_a\cup S_b = \{s_{j^1} \dots s_{j^1+|S_a|-1}, s_{j^2} \dots s_{j^2 + |S_b| -1}\}$. And we proved that:
	$$
		\forall S \in \Gamma_{p+1}, \exists j \in \{1 \dots \mathcal{S}\} \text{ such that } S=\{s_j \dots s_{j+|S|-1}\}
	$$

\end{proof}
Theorem \ref{th:order} implies that there is a total order between the clusters:
$$
	\forall S, S' \in \Gamma_i, S \preceq S' \iff max(S) \le min (S')
$$

Thus, the clusters $\Gamma_i = \{S_1 \dots S_g\}$ can be sorted from the cluster $S_1$ containing the highest scores to $S_g$ containing the lowest scores.

\begin{distances}{The intra-cluster distances are smaller than the inter-cluster distances, i.e.,}
	\label{th:dist}
	\begin{align*}
		\forall i \in & \{0 \dots n-1\},                                                                        \\
		min           & \{D_{min}(S,S'): \forall S,S' \in \Gamma_i, S \neq S' \} \ge                            \\
		              & max\{|s_j - s_{j+1}|: \forall S \in \Gamma_i, \forall s_j \in S \backslash \{min(S)\}\}
	\end{align*}
\end{distances}
\begin{proof}
	We proceed by induction.

	For $i=0$, it is true because $\Gamma_0$ only contains one-point clusters.

	Let us suppose the statement is true for $i=p$ with $p<n-1$ and we will prove that:
	\begin{equation}
		\begin{aligned}
			min & \{D_{min}(S,S'): \forall S,S' \in \Gamma_{p+1}, S \neq S' \} \ge                            \\
			    & max\{|s_j - s_{j+1}|: \forall S \in \Gamma_{p+1}, \forall s_j \in S \backslash \{min(S)\}\}
		\end{aligned}
		\label{eq:dist}
	\end{equation}

	\begin{multline*}
		\text{It exists two clusters } S_a, S_b \in \Gamma_p \text{ such that}\\
		\quad \text{for all }  S,S' \in \Gamma_p, S \neq S',
		D_{min}(S_a, S_b) \le D_{min}(S, S') \\
		\shoveleft{\implies \Gamma_{p+1} = (\Gamma_p \backslash \{S_a, S_b\}) \cup \{S_a \cup S_b\}}
	\end{multline*}

	From Equation \eqref{eq:prop_lvl}, we deduce that:
	\begin{align*}
		\qquad \mathcal{L}(\Gamma_{p+1}) & \ge \mathcal{L}(\Gamma_{p})                                             \\
		\iff  min                        & \{D_{min}(S,S'): \forall S,S' \in \Gamma_{p+1}, S \neq S' \} \ge        \\
		                                 & min\{D_{min}(S,S'): \forall S,S' \in \Gamma_{p}, S \neq S' \}           \\
		\implies min                     & \{D_{min}(S,S'): \forall S,S' \in \Gamma_{p+1}, S \neq S' \} \ge        \\
		                                 & max\{|s_j - s_{j+1}|: \forall S \in \Gamma_{p+1}\backslash \{S_a,S_b\}, \\
		                                 & \qquad \qquad \forall s_j \in S \backslash \{min(S)\}\}
	\end{align*}

	Thus, to prove Equation \eqref{eq:dist}, we only have to show that:
	\begin{align*}
		min & \{D_{min}(S,S'): \forall S,S' \in \Gamma_{p+1}, S \neq S' \} \ge              \\
		    & max\{|s_j - s_{j+1}|: s_j \in (S_a\cup S_b) \backslash \{min(S_a\cup S_b)\}\}
	\end{align*}

	Let us consider $S_a \succeq S_b$ and note $j_{sep}=min(S_a)$. Thus, $s_{j_{sep}+1}=max(S_b)$ because of Theorem \ref{th:order} and $D_{min}(S_a, S_v) = s_{j_{sep}} - s_{j_{sep}+1}$. From our hypothesis, we know that
	\begin{align*}
		D_{min}(S_a,S_b) \ge                   & max\{|s_j - s_{j+1}|: \forall S \in \{S_a,S_b\},        \\
		                                       & s_j \in S \backslash \{min(S)\}\}                       \\
		\implies s_{j_{sep}} - s_{j_{sep}+1} = & max\{|s_j - s_{j+1}|:                                   \\
		                                       & s_j \in (S_a\cup S_v) \backslash \{min(S_a\cup S_b)\}\}
	\end{align*}

	Finally, from Equation \eqref{eq:prop_lvl} again:
	\begin{align*}
		\qquad \mathcal{L}(\Gamma_{p+1}) & \ge \mathcal{L}(\Gamma_{p})                                                           \\
		\iff  min                        & \{D_{min}(S,S'): \forall S,S' \in \Gamma_{p+1}, S \neq S' \} \ge                      \\
		                                 & min\{D_{min}(S,S'): \forall S,S' \in \Gamma_{p}, S \neq S' \}                         \\
		\iff  min                        & \{D_{min}(S,S'): \forall S,S' \in \Gamma_{p+1}, S \neq S' \} \ge                      \\
		                                 & D_{min}(S_a,S_b) = s_{j_{sep}} - s_{j_{sep}+1}                                        \\
		\implies min                     & \{D_{min}(S,S'): \forall S,S' \in \Gamma_{p+1}, S \neq S' \} \ge                      \\
		max\{|                           & s_j - s_{j+1}|: s_j \in (S_a\cup S_b) \backslash \{min(S_a\cup S_v)\}\}               \\
		\implies min                     & \{D_{min}(S,S'): \forall S,S' \in \Gamma_{p+1}, S \neq S' \} \ge                      \\
		max\{|                           & s_j - s_{j+1}|: \forall S \in \Gamma_{p+1}, \forall s_j \in S \backslash \{min(S)\}\}
	\end{align*}

\end{proof}

Theorem \ref{th:dist} can be interpreted as follows: the consecutive distances inside the clusters are smaller than (or equal to) the distances between the clusters, if the distance between two clusters is the dissimilarity $D_{min}$.

\subsection{Finding the best-candidate cluster}
\label{app:main}
\begin{bestcandidate-cluster}
\label{def:smax}
We define the best-candidate cluster $S_{max}$, with MaxSize its maximum size, as:
\begin{align*}
	\exists & i_{max} \in \{0, \dots , n - 2\},                                                          \\
	        & \quad i_{max} = max\{i: \exists S \in \Gamma_i, s_1 \in S \wedge |S| \le \text{MaxSize} \} \\
	        & \; S_{max} \in \Gamma_{i_{max}}, s_1 \in S_{max}
\end{align*}
\end{bestcandidate-cluster}

\begin{lem} The following holds:
	\begin{align*}
		\exists i \le & \text{MaxSize}, S_{max} = \{s_1, \dots ,s_i\} \text{ and }        \\
		              & \forall j \le \text{MaxSize} + 1, s_i - s_{i-1} \le s_j - s_{j-1}
	\end{align*}
\end{lem}
\begin{proof}
	From Theorem \ref{th:order} of Appendix \ref{app:side_th}, we deduce that $\Gamma_{i_{max}} = \{S_1, S_2 \dots  S_g\}$ with $S_1$ and $S_2$ are the two clusters containing the highest scores $S_1 \succeq S_2$. Note that $S_{max}=S_1$. We also have $S_1 \cup S_2 = \{s_1  \dots  s_h\}$. We must have $h > \text{MaxSize}$, otherwise, it would exist $i'>i_{max}$ such that $\Gamma_{i'} = {S_1' \dots S_g'}$ with $|S_1'| \le \text{MaxSize}$. This is impossible because of the Definition \ref{def:smax}.

	Moreover, we know that $\exists j_{sep} \le \text{MaxSize}, |S_{max}|=|S_1|=j_{sep}$. Then, $max(S_2) = s_{j_{sep}+1}$ and $min(S_1)=s_{j_{sep}} \implies D_{min} (S_1,S_2) = s_{j_{sep}} - s_{j_{sep}+1}$

	From Theorem \ref{th:dist} of Appendix \ref{app:side_th}, we deduce that:
	\begin{align*}
		 & D_{min}(S_1, S_2) \ge s_j - s_{j+1}, \forall j < h, j \neq j_{sep}                          \\
		 & \quad \implies s_{j_{sep}} - s_{j_{sep}+1} \ge s_j - s_{j+1}, \forall j < h, j \neq j_{sep} \\
		 & \; \text{However, } j_{sep} \le \text{MaxSize}                                              \\
		 & \quad \implies s_{j_{sep}} - s_{j_{sep}+1} \ge s_j - s_{j+1}, \forall j < \text{MaxSize}
	\end{align*}
	Finally, we obtain:
	\begin{align*}
		\exists i=j_{sep} \le & \text{MaxSize}, S_{max} = \{s_1, \dots ,s_i\} \text{ and }        \\
		                      & \forall j \le \text{MaxSize} + 1, s_i - s_{i-1} \le s_j - s_{j-1}
	\end{align*}
\end{proof}

\section{Probability distributions}
\label{app:prob}
In this section, we consider the probability distributions used to simulate the query distribution. In our experiments, we have $N$ possible values from which we extract $l$ unique elements. The values are trapdoors and $N = m_{\text{real}}$.

\paragraph{Uniform distribution}
The (discrete) uniform distribution is defined over $\{a \dots b\}, a,b \in \mathbb{Z} \text{ and } a \le b$ and is denoted as $\mathcal{U}\{a,b\}$. Its probability mass function is the following:
$$
	\forall k \in \{a \dots b\}, f(x; a, b) = \frac{1}{b - a + 1}
$$

In our case, we want to sample values (i.e., trapdoors) that are indexed from 1 to $N$ (i.e., $m_{\text{real}}$) so we can use  $\mathcal{U}\{1,N\}$. We end up to the following probability mass function:
$$
	\forall k \in \{1 \dots N\}, f(k; 1, N) = \frac{1}{N}
$$
\paragraph{Zipfian distribution}
The Zipfian distribution has two parameter $N$ and $s$. Its probability mass function is the following:
$$
	\forall k \in \{1 \dots N\}, f(k; s, N) = \frac{k^{-1}}{\sum_{n=1}^{N} n^{-1}}
$$

In our case, $k$ is the rank of a trapdoor. The trapdoor for which the underlying keyword is the $k$-th most frequent has the rank $k$. As in IKK, we use $s=1$ which represents the empirical source of the distribution.
\paragraph{Inverted Zipfian distribution}
The inverted Zipfian distribution has two parameters $N$ and $s$. Its probability mass function is the following:
$$
	\forall k \in \{1 \dots N\}, f(k; s, N) = \frac{(N - k + 1)^{-1}}{\sum_{n=1}^{N} n^{-1}}
$$

This distribution simply inverts the rank of the values compared to the Zipfian distribution: the highest rank becomes the lowest rank, the second highest rank becomes the second lowest rank and so on. We also use $s=1$ in this distribution.

\end{document}